\documentclass[12pt]{article} 
\usepackage{enumitem}%%Enables control over enumerate and itemize environments
\setenumerate{itemsep=1pt, topsep=1pt}
\usepackage{setspace} %%Enables \doublespacing command for double linespacing
\expandafter\def\expandafter\quote\expandafter{\quote\onehalfspacing} %%Makes 1.5 quote spacing
\usepackage{graphicx}
\usepackage{authblk}

\usepackage{fancyhdr} %%Permits %\pagestyle{fancy} %%Header style
\usepackage{titlesec} %%Header style
\titlespacing*{\subsection}{\parindent}{.25in}{\wordsep}% Reduces spacing after headings
\usepackage{scrextend} %%Allows for changes to footnotes
\deffootnote[1em]{0in}{1em}{\textsuperscript{\thefootnotemark \ }} %%Footnote style
\usepackage{amssymb, amsmath, mathrsfs, amsthm, braket} %%Math packages
\numberwithin{equation}{section}

\usepackage{stmaryrd} %%Use \llbracket and \rrbracket for double brackets
\usepackage{physics}
\usepackage{parskip}
\usepackage{siunitx}
\usepackage{mathtools}
\usepackage{float}
\usepackage{tensor}
\usepackage{framed}
\usepackage[colorlinks=true, citecolor=blue, linkcolor=blue]{hyperref}
\usepackage{orcidlink}
\usepackage{diagbox}
\usepackage[nameinlink]{cleveref}

\crefname{buckingham}{Buckingham-$\Pi$ Theorem}{Buckingham-$\Pi$ Theorems}
\Crefname{buckingham}{Buckingham-$\Pi$ Theorem}{Buckingham-$\Pi$ Theorems}
\crefformat{buckingham}{#2Buckingham-$\Pi$ Theorem#3}
\Crefformat{buckingham}{#2Buckingham-$\Pi$ Theorem#3}

\crefformat{schwarz}{#2Schwarz's Theorem#3}

\newcommand{\R}{\mathbb{R}}
\newcommand{\Z}{\mathbb{Z}}

\newcommand{\N}{\mathbb{N}}
\newcommand{\C}{\mathbb{C}}
\renewcommand{\dim}{\mathrm{dim}}

\newtheorem{theorem}{Theorem}
\newtheorem{definition}[theorem]{Definition}
\newtheorem{lemma}[theorem]{Lemma}
\newtheorem{proposition}[theorem]{Proposition}
\newtheorem{corollary}[theorem]{Corollary}

\usepackage[round]{natbib} %%Or change 'round' to 'square' for square backers
\usepackage{etoolbox} %%For \citepos
\usepackage{xstring} %%For \citepos

\makeatletter %definition of \citepos
\DeclareRobustCommand\citepos % define \citepos
  {\begingroup
   \let\NAT@nmfmt\NAT@posfmt% same as for citet except with a different name format
   \NAT@swafalse\let\NAT@ctype\z@\NAT@partrue
   \@ifstar{\NAT@fulltrue\NAT@citetp}{\NAT@fullfalse\NAT@citetp}}
   
\let\NAT@orig@nmfmt\NAT@nmfmt %makes adapt to last names ending with an 's'.
\def\NAT@posfmt#1{%
  \StrRemoveBraces{#1}[\NAT@temp]%
  \IfEndWith{\NAT@temp}{s}
    {\NAT@orig@nmfmt{#1'}}
    {\NAT@orig@nmfmt{#1's}}}
\makeatother

\usepackage{csquotes}
\MakeOuterQuote{"}

\makeatletter
\newcommand*{\citelinktext}[2]{%
  \hyper@@link[cite]{}{cite.#1}{#2}}
\makeatother

\begin{document}
\title{\Large\textbf{Dimensional Analysis is a Gauge Theory}}
\author[]{Benjamin Muntz\footnote{\href{mailto:benjamin.muntz@nottingham.ac.uk}{benjamin.muntz@nottingham.ac.uk}}\ \ \orcidlink{0000-0002-0183-8783}}
\affil[]{School of Physics and Astronomy, University of Nottingham, University Park, Nottingham NG7 2RD, United Kingdom}
\affil[]{Nottingham Centre of Gravity, University of Nottingham, University Park, Nottingham NG7 2RD, United Kingdom}
\setlength{\affilsep}{\baselineskip}
\renewcommand\Affilfont{\itshape\small}

\date{\today}
\maketitle
\thispagestyle{empty}

\begin{abstract}
\noindent I argue that dimensional analysis can appropriately be thought of as a gauge theory. This picture naturally leads to the usual quantity calculus through inherent properties of Lie groups, Lie algebras, and representation theory. The gauge theory interpretation is perhaps strongest for non-constant units. I explain how Stevens's classification of scales of measurement can be understood by choices of different gauge groups. Finally, I reinterpret and rephrase the Buckingham-$\Pi$ Theorem in a way that also applies to typical gauge theories. Counting the number of ``$\Pi$''s becomes a well-known problem of invariant theory.
\end{abstract}

%\tableofcontents

%\doublespacing
\onehalfspacing

%%%%%%%%%%%%%%%%%%%%%%%% NOTES %%%%%%%%%%%%%%%%%%%%%%%%%

%%%%%%%%%%%%%%%%%%%%%%% DOCUMENT %%%%%%%%%%%%%%%%%%%%%%%

\section{Introduction}
Consider the following riddle:
\emph{
\begin{quote}
    Vectors live in vector spaces.\\
    Symmetries live in groups.\\
    But units --- the meter, second, gram, and so on --- live where? 
\end{quote}
}
A lot can be intuited from where someone lives. The same is true for mathematical objects. How and why vectors may be added, scaled, or transformed becomes mathematically firmer once we abandon cartoonish arrows and situate them in vector spaces. Likewise, symmetry transformations acquire foundation once we consider them as elements of groups (typically acting on something). But for units, the situation is strangely unclear. From everyday reasoning we understand and picture how one can add one meter to another meter, or one second to another second, yet adding one second to a meter is completely nonsensical. Moreover, multiplying two lengths yields an area, dividing distance by time yields a speed, and so on. There is clearly a coherent algebraic structure at play --- the so-called \emph{quantity calculus} --- but it is nowhere stated where that structure actually originates from. If vectors inherit their algebraic properties from the fact that they inhabit vector spaces, in what space do things like meters, seconds, and grams reside? 

Since the late nineteenth century, physicists and philosophers alike have sought to say precisely what quantities are and what role dimensionality plays in equations of physics and in theory of measurement. Today, discussions more often revolve around understanding the representational status of dimensionful quantities --- whether dimensionality is merely a conventional bookkeeping tool, or instead tracks objective quantitative structure, perhaps only up to equivalence under admissible reparameterisations of base units and dimensions \citep[see, e.g.,][]{Skow2017, Grozier2020, Sterrett2021, Jacobs2024, Jalloh2025}. Untangling all the accounts and long-standing debates would take us too far afield from the purpose of this paper, and, that being said, even my honest physicist attempt would most likely not do the subject proper justice. For that reason I will humbly point the reader to other excellent historical overviews such as \citet{deBoer1995}, \citet{Roche1998}, \citet{Mitchell2017}, and \citet{DeClark2017}.

Let us instead seize the luxury of brevity to jump right into action.

One outstanding problem in dimensional analysis today is the question of structural foundation. While the basic rules of quantity calculus are somewhat uncontroversial (add only like dimensions, multiplication of units adds their exponents, equations must be dimensionally homogeneous), the literature lacks consensus where exactly these properties are mathematically rooted. Several attempts exist: modern takes like \citet{Janyska2007, Janyska2010}, \citet{Tao2012}, \citet{Domotor2017}, \citet{Raposo2018}, and \citet{Zapata2022} have put forward specialised axiomatic models, but these approaches often take some version of the quantity algebra as primitive, rather than explaining it as an instance of a broader structure already familiar from physics. Multiplying mathematical definitions can at times contribute to shrouding physical intuition, leaving the physicist faintly unsatisfied. Dimensional analysis --- a discipline whose strength lies in its closeness to physical intuition --- is by no means immune to this pitfall. Burying dimensional analysis under a mountain of mathematical sediment for purely structural purposes therefore misses the mark; what gain is there when the soil does not nourish the roots? One would hope that an encompassing mathematical foundation ought not merely to recover said intuition, but to illuminate and more importantly \emph{further} it. We need not look far to achieve this. In fact, I hope to persuade the reader that there is no need to reinvent the wheel at all, and that the mathematical foundation for dimensional analysis has already been well-established and employed in physics for many decades. My message is precisely this: 

\vspace{6pt}
\begin{quote}
    \hyphenpenalty=10000
    Dimensional analysis can be wholly understood as a gauge~theory. 
\end{quote}
\vspace{6pt}

The paper is organised as follows: in \Cref{sec:gaugetheory} I detail how dimensional analysis reads as a gauge theory in the language of principal fibre bundles. Since the basic structure involves an abelian group (isomorphic to $\R_+^k$ with $k\in \mathbb{N}$),\footnote{I denote $\R_+$ as the Lie group of positive reals with multiplication as group operation. Some literature may prefer the notation $\R_{>0}$ or $\R^+$.} I will draw several parallels to $U(1)^k$ gauge theory along the way. In \Cref{sec:nonconstant} I explain how the gauge theory description also accommodates units that vary in spacetime --- a convention often discussed in the context of gravitational theories but nevertheless left out in most dimensional analysis literature. In \Cref{sec:dimensionless} I discuss how Stevens's classification of scales of measurement can be accounted for by the choice of structure group, yielding a gauge-theoretic distinction between notions of dimensionless and unit-invariant quantities. Finally, in \Cref{sec:buckingham}, I show how the infamous \cref{thm:Buckingham} can be understood from a gauge theory perspective. Importing it back into the context of gauge theory, I then discuss several examples and applications of the analogous version of the \cref{thm:Buckingham} for general (including non-abelian) gauge theories.

\section{Dimensional Analysis as a Gauge Theory}\label{sec:gaugetheory}
In this section I outline how dimensional analysis can be viewed as a gauge theory. I expect that the reader may find the mathematics presented here suspiciously elementary (in fact, it could just as well be used as a basic reminder of the mathematics of gauge theory). This is meant to reinforce the point that nothing exotic is required; all key features of dimensional analysis and quantity calculus will follow from very basic definitions and properties of fibre bundles and representation theory, paralleling the physicist's way of introducing ordinary gauge theory relevant to particle physics. That being said, my intention is by no means to undermine the core message and perspective that is being conveyed. Notions such as unit `transformations', `\mbox{-independence}', and `choosing a basis of units' become exact mirrors of gauge `transformations', `-invariance', and `-fixing'. Notwithstanding, I daresay most theoretical physicists today find discussions of units much less palatable than discussions of gauge theory. Today there is widespread culture and convention in theoretical physics to set all base units equal to 1 in order to circumvent the hassle of keeping track of units altogether. Yet, we are all taught cautionary tales about how ``gauge-fixing too early'' can obstruct our physical understanding.\footnote{Consider for example classical electromagnetism where $A_0$ plays the role of a Lagrange multiplier imposing Gau\ss's law $\vec{\nabla}\cdot \vec{E} =\rho$. If one picks temporal gauge $A_0=0$ \emph{before} varying the action, this removes the constraint and Gau\ss's law is no longer enforced.} Fixing units from the get-go can in certain cases also lead one astray. An example of this was highlighted in \citet{KaramitsosMuntz2025} in the context of the Swampland Programme, which sparked the idea to pursue this work.

Now, just as the concept of symmetries is inseparable from gauge theory, the role of symmetries in dimensional analysis has always been lurking on the surface since its historical origin. \citet{Fourier1878} recognised early on that dimensional homogeneity should be understood as the requirement that an equation remains valid under changes of units; when a base unit is replaced by a rescaled one, numerical values transform by definite \emph{conversion factors}, and ``dimensions'' encode the exponents governing this transformation \citep{DeClark2017}. This perspective later became central in the nineteenth-century development of dimensional formulae, where these were used to organise unit conversion and systems of measurement. Put this way, changing units is in broad strokes a symmetry of representation. `Conversion factors' are emphasised here because they are arguably more universal than the numerical coefficients of quantities: translating a length expressed in, say, metres to feet requires only a single scale factor that applies uniformly across any numerical input. The freedom to rescale $k$ independent base units can accordingly be encoded in the multiplicative group $\R_+^k$. The presence of a group structure, however, does not in itself lend credence to a gauge theory prescription. Most elementary discussions of dimensional analysis stop at a single, global choice of scale. As I will further argue in \Cref{sec:nonconstant}, there is no reason to impose that restriction from the outset. Imagine a world where Europe and the US were two isolated regions with no communication between their respective residents --- Europeans could go about using metric units completely unaware that their friends across the pond enjoy Imperial units instead. Only upon bridging the gap need we be mindful how to convert our measurements so that we may understand one another. Luckily we do not live in such a world. It does, however, exemplify an important detail: the choice of unit is evidently \emph{local}. Organising the freedom to make such a local choice of scale and compare choices on overlaps is precisely what a principal bundle is for.\footnote{\label{footnote:conformalgeometry}I should also acknowledge that much of my thinking has been influenced by conformal geometry as explained in \citet{CurryGover2017}. Here, an object of interest is ${\mathcal{Q}\subset S^2T^*M}$, the bundle of symmetric $(0,2)$-tensors $g$ on $M$, where all $g$ belong to the same conformal equivalence class $[g]$. That is, $g\sim g'$ if they are related by a Weyl transformation $g'=\Omega^2 g$ where $\Omega^2$ is some positive function on $M$. This has the structure of a principal $\R_+$-bundle (also known as a ray subbundle) $\mathcal{Q}\to \mathcal{G}\to M$ where locally $\mathcal{G}\simeq U\times \{g\}$, i.e. the fibre of $\mathcal{G}$ contains only the representative $g$ of $[g]$. Note that there is no unique choice of $\mathcal{G}$ since one can pick any $\tilde{g}\in [g]$ as the representative. Still, they all lead to isomorphic bundle structures. A `choice of scale' on a conformal manifold can be related to a choice of unit, as also alluded to in \citet{KaramitsosMuntz2025} and later in \Cref{sec:scalartensor}.} This motivates the following basic definition.

\begin{definition}
	The \textbf{scale bundle} is a principal $\R_+^k$-bundle $P\to M$ for some $k\in \N$. 
\end{definition}

For example, for $k=3$, we could label each $\R_+$ by the mechanical base dimensions $\mathrm{M,L,T}$. The right action of $(1,2,1)\in \R_+^3$ sends a frame to one in which the length scale is doubled. Note however that, a priori, $\mathrm{M,L}$, and $\mathrm{T}$ are just names. One could attach physical meaning to each dimension --- e.g.~$\mathrm{L}$ is associated with physical lengths between points on the base manifold --- but that is strictly not necessary as far as dimensional analysis is concerned. As a mere calculational tool, the `practice of dimensional analysis' (by this I mean for instance algebraically manipulating dimensionalities or testing an equation for dimensional homogeneity) is indifferent to what the base dimensions are called or what they physically represent. Moreover, any other choice of base dimensions, say $\mathrm{M,L}$, and $\mathrm{V=LT^{-1}}$, we can just recognise as a change of basis in $\R_+^3$. These arguments generalise to any $k$.

Notice also that the element $(1,2,1)\in \R_+^3$ in the above example does more than just double lengths: it scales areas by four, volumes by eight, and so on. How is this accounted for? The common strategy employed by many models of dimensional analysis, including some mentioned previously in the introduction, is to functorially carry over the appropriate transformation properties by engineering quantities such as areas and volumes as `tensor powers' (integer or rational\footnote{Another common disposition in the dimensional analysis literature is to consider only quantities whose units are raised to rational powers. \citet[p.~71]{Raposo2019} even goes as far as to argue that the reals are ``unnecessarily oversized'' and in fact integers are entirely sufficient. I find this propensity confusing and surmise that this modelling choice stems from the fact that most physics disciplines only ever encounter rational powers in practice (in particular classical mechanics, where there is a long tradition of applying dimensional analysis). But again, as far as dimensional analysis as a calculational tool is concerned, nothing is stopping us from manipulating e.g.~a length to some irrational power. A comment under \citeauthor{Tao2012}'s blog post also rightfully points out that irrational powers do in fact show up in physics: for CFT two-point functions, schematically $\sim (\text{length})^{-2\Delta}$, the scaling dimension $\Delta$ can be non-rational (and is in practice not expected to be rational). Restricting to $\mathbb{Z}$ or $\mathbb{Q}$ thus reads more like artificial aesthetics rather than inferred by some underlying physical principle. The gauge theory picture renders this commitment unnecessary.}) of lengths. This enforces the correct scaling by splitting quantities of different dimensions into a proliferating family of distinct bundles or graded components. Subsequently, one stipulates a coherence principle (functor) that makes the same scale or unit transformation act compatibly on all of them. That is, it manufactures the simultaneity rather than explaining it.\\
\indent Turning to the gauge theory point of view, the interpretation is somewhat different. Now is a good time to make an analogy with $U(1)$ gauge theory. Under the group action by some element $e^{i\theta}\in U(1)$, one field (i.e.~a section of a vector bundle associated with the principal $U(1)$-bundle) may pick up a phase $\phi\mapsto e^{i\alpha\theta}\phi$ and another $\psi\mapsto e^{i\beta\theta}\psi$, with $\alpha\neq \beta$. Any physicist will perhaps appreciate at a basic level that this is because $\phi,\psi$ possess different $U(1)$ \emph{charges}. Put in more mathematical terms, what is really meant is that $\phi$ and $\psi$ transform under different \emph{weighted representations} of $U(1)$. My claim is that it is natural to think similarly about the dimensionality of quantities: two quantities whose dimensionalities are e.g.~$\mathrm{L}^\alpha$ and $\mathrm{L}^\beta$ simply transform under different weighted representations of $\R_+^k$. They have ``different charges'' under scale transformations.\footnote{In the language of \citet{Gomes2024, Gomes2025}, here I am adopting a symmetry-first perspective of gauge theory. Alternatively, one could explore dimensional analysis in a geometry-first formulation, where the symmetry group manifests as the automorphisms of the internal geometry. While it is unclear to me at the present moment whether this presents any merit paralleling \citeauthor{Gomes2024}'s account of the Standard Model, I find the option very interesting.}

That being said, there are subtle differences between dimensional analysis and its $U(1)$ counterpart. Because $\R_+^k$ is contractible, any principal $\R_+^k$-bundle over a paracompact manifold is (topologically) trivial. In other words, any local trivialisation extends to a global trivialisation $P\simeq M\times \R_+^k$. Although we lose some of the interesting structure and properties of fibre bundles, the picture still remains conceptually useful. Perhaps it is even a benefit, as it distinguishes dimensional analysis from other gauge theories we usually encounter in physics: there is no charge quantization, no monopoles/instantons, and representations classify flat connections up to gauge. A particle physicist may likely label dimensional analysis as the most boring gauge theory.

Now, with $G=\R_+^k$ as our Lie group, its Lie algebra is $\mathfrak{g}\simeq \R^k$. Viewing $\mathfrak{g}$ as an additive group, the log map is a smooth Lie group isomorphism.
\begin{equation}
\begin{split}
	\log\colon\qquad \qquad \quad  G & \longrightarrow \mathfrak{g} \\
    (g_1,\dots,g_k) & \longmapsto  (\log g_1,\dots, \log g_k)
\end{split}	
\end{equation}
Charges then live in what is known as the weight space.

\begin{definition}
    The \textbf{weight space} related to the scale bundle is the dual vector space $\mathfrak{g}^*= \mathrm{Hom}_\R(\mathfrak{g},\R)$. A \textbf{weight} $w\in \mathfrak{g}^*$ defines a one-dimensional representation $\rho_w\colon \R_+^k \to \R_+\subset\mathrm{GL}(1,\R)$ given by $\rho_w(g) = \exp(w(\log g))$ for $g\in \R_+^k$.
\end{definition}
Consider again the example where $k=3$. The representation corresponding to the weight $w=(w_1,w_2,w_3)\in \mathfrak{g}^*\simeq \R^3$ maps $g=(g_1,g_2,g_3)\in \R_+^3$ to 
\begin{equation}
    \rho_w(g) = e^{w(\log g)} = g_1^{w_1}g_2^{w_2}g_3^{w_3}
\end{equation}
which is the appropriate scale factor associated to a quantity with dimensionality $(w_1,w_2,w_3)$.

Returning to the case of $U(1)$ (and also other non-trivial Lie groups), one has to be slightly more careful in identifying the set of possible charges. Although any $\lambda\in \mathfrak{g}^*$ defines a Lie algebra representation, only some exponentiate to global characters $\chi\colon U(1)\to \C^\times$. Compactness implies a bounded image, and thus only the subset $\mathfrak{w}^*\subset \mathfrak{g}^*$ that is integral on $\mathrm{ker}(\exp\colon \mathfrak{g}\to G)$ is compatible with the global topology. For $U(1)$ the weights lie on a lattice $\mathfrak{w}^*\simeq \Z$ which is why charges become quantized. For $\mathbb{R}_+$ the exponential map is a global diffeomorphism with trivial kernel, so any such $\lambda$ exponentiates to a global character, i.e. $\mathfrak{w}^*=\mathfrak{g}^*$.

The fact that charges effectively add up under multiplication then follows from basic representation theory.
\begin{lemma}\label[lemma]{lem:tensorproductrep}
    Let $w,w'\in \mathfrak{g}^*$ and $\rho_w,\rho_{w'}$ be the corresponding one-dimensional representations. Then $\rho_w\otimes \rho_{w'}\simeq \rho_{w+w'}$.
\end{lemma}
\begin{proof}
    For $g\in G$ we can simply compute
    \begin{equation}
        (\rho_w\otimes \rho_{w'})(g) = \rho_{w}(g)\rho_{w'}(g) = e^{w(\log g)}e^{w'(\log g)} = e^{(w+w')(\log g)} = \rho_{w+w'}(g)\,.
    \end{equation}
\end{proof}

In the previous discussion I motivated how one can get away with thinking about quantities like charged fields in ordinary gauge theory. Let us intuit this structural equivalence from another perspective by considering \citeauthor{Poincare1906}'s famous thought experiment. For the sake of conciseness I provide here only a compact version, however a more elaborate account can be found in \cite{Poincare1906} \citep[see also][]{Jalloh2025}: imagine you wake up one morning in a world where every physical realisation of length has uniformly become a thousand times greater, with no additional mark left behind. Would you notice? Empirically speaking, the answer to this question appears to be no, since every comparative statement producing dimensionless observables would return the exact same output. \citet{Jalloh2025} labels this as exhibiting an \emph{ontic symmetry}: there is a transformation which changes the numerical representatives of quantities for any given unit system while leaving their ratios unchanged. Descriptions related by such a transformation should therefore be treated as equivalent, insofar as the transformation generates an empirically indistinguishable system. Notwithstanding, once we turn to the language of fibre bundles, \citeauthor{Poincare1906}'s thought experiment receives a straightforward reinterpretation. Imagine (as an inhabitant on the scale bundle) you wake up one morning and overnight someone has performed a right action of an element $g\in \mathbb{R}_+^k$, sending $p\mapsto p\cdot g$ for an original choice of frame $p\in P$. If $q\in \mathbb{R}$ denotes a putative numerical representative of a weight-$w$ quantity, then the overnight transformation $p\mapsto p\cdot g$ is equivalently representable by a corresponding rescaled number $\rho_w(g)q$ (an active scale transformation). Using some suggestive notation, allow me to write this identification or equivalence relation as
\begin{equation}
    [p\cdot g,q]\sim [p,\rho_w(g)q]\,.
\end{equation}
That is just the intrinsic property of an associated bundle!
\begin{definition}\label[definition]{def:quantity}
    Given the weighted representation $\rho_w$ we define the associated line bundle $\mathcal{L}_w = P\times_{\rho_w}\mathbb{R}$. A \textbf{quantity of weight} $w$ is a section $Q\in \Gamma(\mathcal{L}_w)$.
\end{definition}
The definition above considers quantities as $\mathbb{R}$-valued sections. This of course does not accommodate all of physics. The velocity vector, Maxwell tensor, Dirac spinor, and so on, are not scalars but can nevertheless be dimensionful. Rest assured that the definition is easily generalised. It goes as follows: let $\mathcal{V}$ be an $\R$-vector space (or more generally over any field $\mathbb{F}\supseteq\R$) and $\mathcal{V}M$ its vector bundle over the base manifold $M$. Define $\mathcal{V}$-valued quantities of weight $w$ as sections of $\mathcal{V}M\otimes \mathcal{L}_w$. As such, it entirely suffices to work at the level of $\mathcal{L}_w$. I will therefore stick to real quantities for simplicity and to stay close to default discussions of dimensional analysis. 

In the global trivialization $P\simeq M\times \R_+^k$, a section of $\mathcal{L}_w$ is equivalent to an equivariant function $q \colon P\to \R$ such that $q(p\cdot g) = \rho_w(g)^{-1}q(p)$ for all $p\in P$ and $g\in \R_+^k$. When $w=0$ we say that the quantity is \textbf{dimensionless} (invariant under the fibre action).\footnote{Note that a quantity being dimensionless does not necessarily imply that it is unit-invariant. This subtlety is discussed in more detail in \Cref{sec:dimensionless}.}

\begin{lemma}\label[lemma]{lem:equivariantfunction}
    Consider the scale bundle $P\to M$ and let $\rho_w$ denote any one-dimensional representation as above. The space of smooth equivariant functions is isomorphic to the space of smooth sections $\Gamma(\mathcal{L}_w)$.
\end{lemma}

\begin{proof}
    Given an equivariant $q\colon P\to \R$ define $Q(\pi(p)) = [p,q(p)]$ in the associated bundle; equivariance ensures the element $[p,q(p)]$ is independent of the choice of representative in the fibre. Conversely, any representative of a section in a trivialization gives an equivariant function.
\end{proof}

We do not have to look far to realise that quantities defined as above obey the well-known algebraic rules of dimensional analysis. They again follow from basic representation theory.

\begin{proposition}
    Properties of quantity calculus follow naturally.
    \begin{enumerate}[label=\alph*)]
        \item If $Q,Q'$ are quantities of weight $w,w'$, then $QQ'$ is a quantity of weight $w+w'$.
        \item If $Q,Q'$ are quantities of weight $w$, then $Q+Q'$ is a quantity of weight $w$.
        \item One can only add/subtract quantities of equal weight (dimensional homogeneity).
    \end{enumerate}
\end{proposition}
\begin{proof}
    Write $Q=[p,q(p)]$ and $Q'=[p,q'(p)]$. By invoking \Cref{lem:equivariantfunction} we can work at the level of the equivariant functions $q,q'$. One by one:
    \begin{enumerate}[label=\alph*)]
       \item Suppose $Q\in \Gamma(\mathcal{L}_w)$ and $Q'\in \Gamma(\mathcal{L}_{w'})$.
    \begin{equation}
    \begin{split}
        q(p\cdot g) q'(p\cdot g) &= \rho_w(g)^{-1}\rho_{w'}(g)^{-1} (q(p) q'(p))\\
        &= \rho_{w+w'}(g)^{-1} ( q(p) q'(p))\\
    \end{split}
    \end{equation}
    Hence $(qq')(p) \coloneqq q(p) q'(p)$ corresponds to a quantity $QQ'\in \Gamma(\mathcal{L}_{w+w'})$.
    \item Bilinearity in $\R$ as a real vector space carries over to bilinearity in $\mathcal{L}_w$. Let $a,b\in \R$ and $Q,Q'\in \Gamma(\mathcal{L}_w)$.
    \begin{equation}
    \begin{split}
        aq(p\cdot g)+bq'(p\cdot g) &= \rho_w(g)^{-1}(aq(p)+bq'(p))
    \end{split}
    \end{equation}
    Thus $aQ+bQ'\in \Gamma(\mathcal{L}_w)$.
    \item Homogeneity follows from the fact that quantities of different dimensionality simply live in different spaces. If $Q\in \Gamma(\mathcal{L}_w)$ and $Q'\in \Gamma(\mathcal{L}_{w'})$ with $w\neq w'$, addition between them is mathematically ill-defined.
    \end{enumerate}
\end{proof}

There is also a more operational way to understand dimensionful quantities from the gauge-theoretic perspective. Consider for instance $c$, the speed of light. Viewed as a quantity in its own right, with non-zero weights $w_\mathrm{L}=1$ and $w_\mathrm{T}=-1$, we have demonstrated that it (and any other speed for that matter) can be interpreted as a particular section of a weighted line bundle. But the speed of light also serves another purpose; as often brought to bear in relativistic physics, it can equivalently be viewed as a map $t\mapsto ct$ sending time intervals to length intervals, or vice versa $x\mapsto x/c$ length into time. This is close to what \citet{Jacobs2022} has described as an \emph{inter-quantity relation}: a dimensionful constant articulates how distinct families of quantities (time and length in the example involving $c$) are related in a way that is stable under changes of scale. Weighted line bundles manifestly embody this intuition.

\begin{corollary}
    An equivalent way to view quantities is as homomorphisms between weighted line bundles, since homomorphisms $\mathcal{L}_{w}\to \mathcal{L}_{w+w'}$ correspond to pointwise multiplication by a section of $\mathcal{L}_{w
        '}$.
\end{corollary}

It is important to remark that any trivialisation $\varphi\colon M\to P$ of the scale bundle naturally induces a canonical section $\mu = [\varphi,1]$ in the associated bundle
\begin{equation}
    \mathcal{L}_w^+ \coloneqq P\times_{\rho_w}\R_+\,.
\end{equation}
In other words, $\mu$ can be viewed as a strictly positive weight-$w$ quantity. Notice now the isomorphism
\begin{equation}
    \Gamma(\mathcal{L}_w)\simeq {}^{\displaystyle \left(C^\infty(M,\R)\times \Gamma(\mathcal{L}_w^+)\right)} {\big\slash}_{\displaystyle \sim}
\end{equation}
where the quotient identifies $(\Pi,\mu)\sim (\lambda^{-1}\Pi,\lambda\mu)$ for $\lambda\in C^\infty(M,\R_+)$. I highlight this isomorphism for the following observation: provided a ${\mu \in \Gamma(\mathcal{L}_w^+)}$ (for instance induced by a trivialisation $\varphi$), any quantity $Q\in\Gamma(\mathcal{L}_w)$ can be written uniquely as
\begin{equation}
    Q =\Pi \mu
\end{equation}
for a unique smooth function $\Pi \in C^\infty(M,\R)\simeq \Gamma(\mathcal{L}_0)$. Moreover, because $\mu$ is everywhere non-vanishing, the map is invertible; the ratio $\Pi=Q/\mu$ is also well-defined. The reader may recognise this as the familiar metrological picture in which a quantity is expressed as a number-unit pair! Writing $Q=[p,q]$, the numerical coefficient $\Pi$ is related to the equivariant $q$ by\footnote{\label{footnote:Pitransforming}Note that, while $\Pi$ is a function on $M$ rather than $P$, it transforms under $\varphi\mapsto \varphi g$ as $\Pi(x) \mapsto \rho_{w}(g)^{-1}\Pi(x)$ due to equivariance of $q$. This is just the passive transformation ensuring $Q=\Pi\,\mu$ is invariant under the group action.}
\begin{equation}\label{eq:Pi_definition}
    \Pi = q\circ \varphi\,.
\end{equation}
In fact, because the associated bundle $\mathcal{L}_w^+$ is just the extension of the structure group of $P$ along the homomorphism $\rho_w\colon \R_+^k\to \R_+$, it is also itself a principal $\R_+$-bundle, with right action given by scale factors $[p,q]\cdot \rho_w(g) = [p,\rho_w(g)q]$. For $k=1$ and $w\neq 0$, $\rho_w$ is an isomorphism, so $P\simeq \mathcal{L}_w^+$ is just the scale bundle itself; the structure group has been relabelled through $\rho_w$. Moreover, sections $\mu\in\Gamma(\mathcal{L}_w^+)$ become trivialisations of $P$.

Finally, we are able to return to the riddle posed in the introduction of this paper: 
\vspace{6pt}
\begin{quote}
    \centering \emph{Where do units live?}
\end{quote}
\vspace{6pt}
We now have the answer:
\vspace{6pt}
\begin{quote}
    \centering A unit is a section $\mu\in \Gamma(\mathcal{L}_w^+)$.
\end{quote}
\vspace{6pt}
Or equivalently, for $k=1$, a trivialisation of the scale bundle --- units are \emph{choices of gauge}.

For $k>1$, one has to be slightly more careful. In this case $\rho_w$ has non-trivial kernel and $\mathcal{L}_w^+$ remembers only a one-dimensional quotient of the full scale bundle. I.e.~$P\not\simeq \mathcal{L}_w^+$ for any $w$ if $k>1$. If we wanted to reconstruct the scale bundle, we would need $k$ non-degenerate associated bundles. The picture is again familiar from dimensional analysis: suppose we have $k$ base dimensionalities (e.g.~$\mathrm{M},\mathrm{L},\mathrm{T},\dots$), we need to specify $k$ linearly independent base units (e.g.~$\SI{1}{\kg},\SI{1}{\meter},\SI{1}{\second},\dots$) in order to parameterise any dimensionful quantity. When there are either fewer units than dimensions or they are degenerate, this is clearly not possible. With more units than dimensions it is on the other hand possible, but the decomposition is no longer unique \citep{Jacobs2024}.

\begin{proposition}\label[proposition]{prop:scalebundleisomorphism}
    Let $P\to M$ be the scale bundle and $w_1,\dots,w_n\in \mathfrak{g}^*$ for some $n\in\N$. The map
    \begin{equation}\label{eq:isomorphism}
    \begin{split}
        \mathsf{\Phi}\colon\quad P & \to \mathcal{L}_{w_1}^+\times_M\cdots \times_M\mathcal{L}_{w_n}^+\\
        p & \mapsto ([p,1],\dots,[p,1]) 
    \end{split}
    \end{equation}
    is a smooth bundle morphism. It is an isomorphism iff $\{w_1,\dots,w_n\}$ forms a complete basis of $\mathfrak{g}^*$. Equivalently, iff the character map
    \begin{equation}
    \begin{split}
        \chi\colon\quad G & \to \R^n_+\\
        g & \mapsto (\rho_{w_1}(g),\dots,\rho_{w_n}(g))
    \end{split}
    \end{equation}
    is a Lie group isomorphism.
\end{proposition}

\begin{proof}
    Compute
    \begin{equation}
    \begin{split}
        \mathsf{\Phi}(pg) &= ([pg,1],\dots,[pg,1])\\
        &= ([p,\rho_{w_1}(g)],\dots,[p,\rho_{w_n}(g)])\\
        &= \mathsf{\Phi}(p)\cdot \chi(g)
    \end{split}
    \end{equation}
    where the last equality indicates the character acting diagonally from the right. So $\mathsf{\Phi}$ is $G$-equivariant after transporting the $G$-action through $\chi$. If $\chi$ is an isomorphism, then $\mathsf{\Phi}$ becomes a bundle morphism between two principal $G$-bundles.\\
    Next we need to show that this is equivalent to saying $\{w_1,\dots,w_n\}$ must be a complete basis of $\mathfrak{g}^*$. Consider therefore the map $W\colon \mathfrak{g}\to \R^n$ with $W(X)\coloneqq (w_1(X),\dots,w_k(X))$ such that $\chi = \exp\circ W\circ \log$. Since $\log$ and $\exp$ are diffeomorphisms, $\chi$ is an isomorphism iff $W$ is a linear isomorphism. But that is only the case iff $\{w_1,\dots,w_n\}$ is a complete basis of $\mathfrak{g}^*$. (And so necessarily $n=k=\dim\,\mathfrak{g}^*$.) This completes the proof.
\end{proof}

Provided $\{w_1,\dots,w_k\}$ forms a complete basis of $\mathfrak{g}^*$, any trivialisation $\varphi$ of $P$ corresponds, via \Cref{prop:scalebundleisomorphism}, to a section 
\begin{equation}
    \mu = (\mu_1,\dots,\mu_k)\in \Gamma(\mathcal{L}_{w_1}^+\times_M\cdots \times_M\mathcal{L}_{w_k}^+)
\end{equation}
with $\mu_i=[\varphi,1]$ for all $i=1,\dots,k$. In more basic terms we could call $(\mu_1,\dots,\mu_k)$ a choice of \textbf{base units}.
It follows that we can always decompose any dimensionful quantity in terms of a choice of base units.
\begin{corollary}\label[corollary]{cor:unitdecomposition}
    Let $\{w_i\}_{i=1,\dots,k}$ be a complete basis of $\mathfrak{g}^*$ and $\{\mu_i\}_{i=1,\dots,k}$ be a choice of base units. For every $v\in \mathfrak{g}^*$ and every quantity $Q\in \Gamma(\mathcal{L}_v)$ there exists a unique set of real numbers $r_i\in \R$ and a unique smooth function $\Pi\in \Gamma(\mathcal{L}_0)\simeq C^\infty(M,\R)$ such that
    \begin{equation}
        Q(x) = \Pi(x) \prod_{i=1}^k (\mu_i(x))^{r_i}\,.
    \end{equation}
\end{corollary}
Overall, we have shown how dimensionalities, units, quantity calculus, and the number-unit decomposition immediately manifest once we recognise dimensional analysis as a gauge theory. It is useful to summarise some important notions:
\begin{description}[font=\textsf,labelindent=1cm] \label{descr:basis}
    \item[Base dimensionalities:] a choice of basis $\{w_i\}$ on $\mathfrak{g}^*$.
    \item[Base units:] a section of the product bundle $\prod_i\mathcal{L}_{w_i}^+$ or, equivalently, a trivialisation of the scale bundle.
\end{description}
Let me thereby highlight that imposing a set of base units is (by following the isomorphism \eqref{eq:isomorphism}) precisely the same as choosing a gauge. A global basis of units always exists because $P$ is trivial.

Within this construction, base dimensionalities are logically prior to base units in a precise structural sense (echoing \citet{Sterrett2021}). To specify base units $\mu_i\in\Gamma(\mathcal{L}_{w_i}^+)$ already presupposes weights $w_i$, and these sections determine a trivialisation of the scale bundle only when the $\{w_i\}_{i=1,\dots,k}$ form a basis of $\mathfrak g^*$. Thus a system of base units is a choice of sections associated with a prior choice of base dimensionalities. 

Let me close this section with one brief comment on related mathematical accounts of dimensional analysis. It is intentional that I have not attempted a systematic comparison with such approaches here. Still, it is worth acknowledging that other authors have not been very far from the structure used above. \citet{Tao2012} and \citet{Domotor2017}, for instance, both motivate and make contact with the notion of $G$-torsors as an appropriate mathematical framework. From the present point of view, this is not particularly surprising: as \citet{Baez2009} explains beautifully, the fibre of a principal $G$-bundle \emph{is} a $G$-torsor. In that sense, while the novel idea of this paper is to thoroughly phrase dimensional analysis as a gauge theory, some of the mathematical ingredients have already appeared in parallel approaches.

For the remainder of the paper, we will explore the conceptual consequences of this account as well as more grounded comparisons with ordinary gauge theory and particle physics.

\subsection{Analogy: $U(1)^k$ gauge theory}\label{sec:U(1)}
Consider $U(1)^k$ gauge theory, which, as a close cousin to $\R_+^k$, concerns an abelian Lie group of dimension $k$. The fundamental object of interest here is the principal $U(1)^k$-bundle, denoted $P^{U(1)}\to M$, from which we can build things such as complex associated line bundles $\mathcal{C}_w=P^{U(1)}\times_{\rho_w} \C$. Saying that $\Psi$ is a complex scalar of charge $w$ just means that it is a section $\Psi\in \Gamma(\mathcal{C}_w)$. This is of course analogous to \Cref{def:quantity} of a dimensionful quantity. (Recall however that the charge in this case is a vector $w=(w_1,\dots,w_k)\in \Z^k$, with quantization due to the compactness of $U(1)^k$.) Likewise, the complex scalars $\Psi$ are in one-to-one correspondence with equivariant functions $\psi\colon P^{U(1)}\to \C$. For some $g = (e^{i\theta_1},\dots, e^{i\theta_k}) \in U(1)^k$ the function transforms as
\begin{equation}
    \psi(pg)= \rho_w(g)^{-1}\psi(p) = e^{-i(w_1\theta_1+\cdots + w_k\theta_k)}\psi(p)\,.
\end{equation}
Now, if $\psi$ has charge $w$ and $\psi'$ has charge $w'$, then the product $\psi\psi'$ transforms as a field with charge $w+w'$. This is again just the statement that tensoring one-dimensional representations adds the charges. Likewise, the sum $\psi+\psi'$ transforms covariantly if and only if $w=w'$; otherwise a gauge transformation rotates the two summands out of sync and the sum does not transform by a single overall phase. Dimensional homogeneity is the equivalent statement that substitutes phases for rescaling factors: a linear combination of quantities is `gauge-covariant' if and only if the summands have equal dimensionality.

Given such a complex scalar charged under $U(1)^k$, it is common to see it represented (locally) by a magnitude and a phase
\begin{equation}
    \psi = \varrho\, e^{i\vartheta}\,.
\end{equation}
We understand that $\varrho = \abs{\psi}$ is gauge-invariant and the phase $e^{i\vartheta}$ transforms under the $U(1)^k$ action. In other words, all we have done is split $\psi$ into a `chargeless' and `charged' piece, akin to how we have seen quantities being written as number-unit pairs
\begin{equation}
    Q = \Pi\, \mu\,.
\end{equation}
Analogously, one could say that the coefficient $\Pi$ and unit $\mu$ are respectively `chargeless' and `charged' with respect to the group of scale transformations $\R_+^k$. Allow me to expand further on this correspondence. Although it is rarely ever useful to invoke in particle physics, for $U(1)^k$-charged gauge fields there is a way to decompose the phase into its `base phases'. Analogous to \Cref{prop:scalebundleisomorphism}, the map
\begin{equation}
    \begin{split}
        \mathsf{\Phi}^{U(1)}\colon\quad P^{U(1)} & \to \mathcal{C}_{w_1}^{U(1)}\times_M\cdots \times_M\mathcal{C}_{w_k}^{U(1)}\\
        p & \mapsto ([p,1],\dots,[p,1]) 
    \end{split}
\end{equation}
where $\mathcal{C}_{w}^{U(1)}\coloneqq P^{U(1)}\times_{\rho_w}U(1)$, is also a principal bundle morphism and an isomorphism if $\{w_1,\dots,w_k\}$ is a complete basis of the weight space $\mathfrak{w}^*\simeq \Z^k$. It tells us that there is a freedom to linearly decompose the phase of $\psi$
\begin{equation}
    e^{i\vartheta} = \prod_{i=1}^k \left(e^{i\vartheta_i}\right)^{r_i}
\end{equation}
such that, under a $U(1)^k$ gauge transformation, each base phase acts like a charge-$w_i$ object $e^{i\vartheta_i}\mapsto \rho_{w_i}(g)^{-1}e^{i\vartheta_i}$. Writing
\begin{equation}
    \psi = \varrho\prod_{i=1}^k \left(e^{i\vartheta_i}\right)^{r_i}
\end{equation}
with integer $r_i$ is just the $U(1)^k$ analogue of \Cref{cor:unitdecomposition} which casts a dimensionful quantity in terms of a set of base units. Note again that the choice of basis is not unique: any other set of base phases $\{e^{i\vartheta'_i}\}_{i=1,\dots,k}$ with weights $\{w_i'\}_{i=1,\dots,k}$ for which
\begin{equation}
    e^{i\vartheta} = \prod_i \big(e^{i\vartheta_i'}\big)^{r_i'} = \prod_i \big(e^{i\vartheta_i}\big)^{r_i}
\end{equation}
does the same job. The only difference is that this can only be done locally, as $U(1)^k$ has non-trivial topology; a local trivialisation does not ascend uniquely to a global one. Because the phases have non-vanishing norm, the ratio $\varrho = \psi / e^{i\vartheta}$ is also well-defined and isolates the gauge-invariant part of $\psi$. Analogously, dividing by the unit extracts the scale-invariant $\Pi = Q/\mu$. 

\section{On Non-Constant Units}\label{sec:nonconstant}
A major practical concern in metrology is the constancy of base units for ensuring replicability of precision measurements and comparisons across experimental settings: when measuring lengths for instance, we would like our reference meter stick to retain its `one-meter-ness' over long periods of time or under different laboratory conditions (temperature, pressure, and so on). That not being the case can blur the physical situation at hand when it comes to interpreting measurement outcomes. Consider a simple thought experiment to illustrate this point. Suppose I wish to measure the mass of a metal cube using a particular glass of water as my reference unit; `one glass' is my standard of mass. The next day I repeat the same measurement and find that the cube's mass has increased. Presuming the cube was safely stored and left untouched, could it really be the case that it magically gained mass overnight? Of course not. More realistically, some of the water from the glass of water had evaporated overnight, increasing the relative mass ratio. Maybe I got thirsty and took a sip? The apparent change in measurement outcome is simply an artefact of having chosen a reference whose `one-glass-ness' is not stable over time. I could in fact also have run the inverted experiment --- measure the mass of the glass of water using the metal cube as a reference unit --- and concluded that the glass has lost mass the next day.\footnote{The two conclusions manifest from the mathematical fact that, for a ratio of two unknown functions $\frac{f_1}{f_2}$ with $\partial(\frac{f_1}{f_2}) \neq 0$, it is impossible to tell from this information alone whether the variation is due to $f_1,f_2$ or both simultaneously. Further deduction requires external insight. For instance, if physical reasoning suggests that $f_2$ does not vary much, we can justify that deviations should for the most part be governed by $f_1$.} This second convention more accurately fits our physical understanding of the situation at hand because the metal cube is the more stable standard. It does not make it illegal to use the glass of water as a measure of mass; it just implies we have to be sufficiently careful and aware about comparing measurement results made at different points in space and time.

The historical development of the meter has been, in effect, an attempt to stop measuring against evaporating glasses of water. From 1799, the \emph{M\`etre des Archives} platinum bar was used as the material standard for defining the meter, later replaced by the international prototype meter in 1889, which was kept under controlled conditions. \citep[For a concise timeline, see][Appendix 4.]{BIPM2019} Naturally, material standards are vulnerable to small variations over time. They can be scratched, contaminated, stressed, or change subtly under handling and aging. Before 1983, the speed of light $c$ measured in $\SI{}{\meter\slash\second}$ would shift over time due to drifts in metrological standards. This is no longer the case, since the meter is now \emph{defined} in terms of the distance light travels in a vacuum during a specified fraction of a second. Does this mean the speed of light was a function of time until precisely 1983, after which it magically became constant? Of course not. The apparent change in measurement outcome was, like the metal cube and glass of water thought experiment, an artefact of having chosen a reference unit that was not stable over time. 

The discussion so far is by no means to advertise for a particular system of base units as the `most stable' or `correct' one to adopt. In fact, it should not matter at all whether one adopts a constant or non-constant system of units (that is a human choice after all), as long as one knows how to subtract the possible phantom contributions due to how the standard varies in space and time. My intention is rather to point out that non-constant units evidently play a significant role in real-life metrology, and any mathematical foundation of dimensional analysis ought to build them in from the outset. Notwithstanding, many discussions of dimensional analysis in the physics and philosophy literature tacitly assume constant units. There are at least two ways one may assert this:
\begin{enumerate}
    \item Conduct dimensional analysis on the dummy manifold $M=\{\text{pt}\}$.
    \item Take the stance that genuine units \emph{must} be constant.
\end{enumerate}
I think it is fair to claim that traditional dimensional analysis takes place on the point where, \emph{ipso facto}, the distinction between constant and varying units is invisible. Metrology on an extended manifold is on the other hand much less straightforward; questions such as ``\emph{is our notion of a meter at $x\in M$ the same as our notion of a meter at $x'\in M$}?'' become tricky to address (c.f.~the history of the meter). In fact, we recognise the same situation when needing to compare, say, vectors on a manifold: comparing a vector $v\in T_xM$ with another vector $v'\in T_{x'}M$ cannot be done from afar. That is, one requires a method to parallel transport $v$ and $v'$ to the same tangent space and manipulate them there. Applying the same logic to dimensional analysis, what we require is some notion of a `connection' to instruct how to `parallel transport' units and quantities between spacetime points. It should not come as a surprise that the formalism laid out here naturally takes care of this --- connections are after all the bread and butter of gauge theory!

In the language of principal bundles, the standard description of parallel transport relies on defining a Lie algebra-valued one-form, $\omega\in \Omega^1(P,\mathfrak{g})$, known as an Ehresmann connection. I will not dive deep into its precise definition or properties; the interested reader can consult \citet{KobayashiNomizu1963} or \citet[Appendix A]{Gomes2024}. For our purposes a simple picture is sufficient: intuitively, one can imagine that parallel transport on $P$ takes place on a distribution of `horizontal spaces' $H\subset TP$ that fibre-wise look like the tangent space of the base manifold, $H_p\simeq T_{\pi(p)}M$. The splitting of $T_pP$ into its horizontal and complementary `vertical' subspace is a priori arbitrary --- it is the Ehresmann connection which assigns a horizontal subspace to each point $p\in P$ in a way that is compatible with the $G$-action on $P$. 

More important for our purposes is that the connection $\omega$ induces a covariant derivative on the associated bundles. Given any trivialisation $\varphi$, we can pull back $\omega$ to obtain a Lie algebra-valued one-form on the base manifold, $A^{(\varphi)}\coloneqq \varphi^* \omega$. (The physicist may intuit this as the gauge potential.) Then for any $v\in\mathfrak{g}^*$ and $Q=\Pi\, \mu^{(\varphi)}\in \Gamma(\mathcal{L}_v)$ in the associated bundle, expressed in the frame $\mu^{(\varphi)}=[\varphi,1]$,
\begin{equation}
    \nabla^\varphi Q=(\dd\Pi + v(A^{(\varphi)})\Pi )\otimes \mu^{(\varphi)} \in \Omega^1(M,\mathcal{L}_v)
\end{equation}
is the coordinatised definition for the covariant derivative $\nabla^\varphi$. In particular, we say that $\omega$ is \emph{adapted} to $\varphi$ if $A^{(\varphi)}=0$. 

Specialising to dimensional analysis, we already saw, following \Cref{prop:scalebundleisomorphism}, what it means to have a trivialisation $\varphi$ of the scale bundle: upon fixing a set of base dimensionalities $\{w_i\}_{i=1,\dots,k}\subset \mathfrak{g}^*$, $\mathsf{\Phi}(\varphi)=(\mu_1,\dots,\mu_k)$ gives a set of base units in the associated bundles, $\mu_i\in\Gamma(\mathcal{L}_{w_i}^+)$. Suppose therefore that we chose a connection adapted to $\varphi$. Carrying it through $\mathsf{\Phi}$ implies that $\nabla^\varphi\mu_i=0$ on each $\mathcal{L}_{w_i}^+$. In other words, adapting the connection to a trivialisation $\varphi\colon M\to P$ is precisely equivalent to \emph{choosing} which set of base units are `covariantly constant'. If $Q=\Pi\,\mu^{(\varphi)}$ is expressed in this unit system,
\begin{equation}
    \nabla^\varphi Q = \dd \Pi \otimes \mu^{(\varphi)}\,.
\end{equation}

Let me emphasise this point clearly. There is no privileged system of units that is intrinsically or by definition ``constant''. I take here the point of view that variation presumes a fixed background structure and, as such, before determining any rate of change, we must first articulate \emph{what} we root ourselves against. This is what adapting the connection to some trivialisation $\varphi$ achieves: it determines the set of units $\mathsf{\Phi}(\varphi)=(\mu_1,\dots,\mu_k)$ that we treat as constant via the property that $\nabla^\varphi\mu_i = 0$. Put bluntly, no unit varies against itself. In \Cref{sec:gaugetheory} we found that dimensionalities are logically prior to units; prolonging that line of thought, one can add that units are logically prior to \emph{constant} units. 

So let us ask the question: what exactly happened in 1983, when the speed of light ceased to vary and truly `became a constant'? On a bundle reading, nothing changed about the speed of light section $c\in \Gamma(\mathcal{L}_{w_c})$. The answer is simply that we chose to work with a connection adapted to a different trivialisation; one that implies $\nabla^\varphi c=0$.\\
The same section can appear to be constant in one choice of background units and vary in another, without anything having changed about the underlying quantity. To see this, suppose we pass from $\varphi$ to another section $\varphi^g=\varphi\cdot g$ for some $g\colon M\to G$. This gives rise to a different set of base units
\begin{equation}
    \mu^{\varphi^g}_i = \mu^{\varphi}_i \cdot \rho_{w_i}(g) = [\varphi,\rho_{w_i}(g)]\,.
\end{equation}
Under the group action the connection transforms via\footnote{In general $(\varphi^g)^* \omega = \mathrm{Ad}(g)^{-1}\varphi^*\omega + g^{-1}\dd g$. As the group is abelian, the adjoint yields the identity. The second term comes from observing that $\log\circ \chi = W\circ\log$ (c.f.~\Cref{prop:scalebundleisomorphism}). Differentiating gives $g^{-1}\dd g = W^{-1}(\dd \log \chi(g))$.}
\begin{equation}\label{eq:connection}
    A^{(\varphi^g)} = A^{(\varphi)} + W^{-1}(\dd \log \rho_{w_1}(g),\dots,\dd \log\rho_{w_k}(g))\,.
\end{equation}
Consider then the same weight-$v$ quantity written in terms of this other set of (non-constant) units $Q = \Pi\,\mu^{(\varphi)}=\Pi^g\, \mu^{(\varphi^g)}\in \Gamma(\mathcal{L}_v)$, with ${\Pi^g=\rho_v(g)^{-1}\Pi}$ (c.f.~\cref{footnote:Pitransforming}). Moreover, decompose $v=\sum_i r_iw_i$ in terms of the base dimensionalities. The covariant derivative now picks up a contribution from the fact that the unit itself is no longer constant (more precisely, the conversion factor relating it to the background which the connection is adapted to)
\begin{equation}
    \nabla^\varphi Q = (\dd \Pi^g + \Pi^g \sum_i r_i\dd \log\rho_{w_i}(g)) \otimes \mu^{(\varphi^g)}\,.
\end{equation}
Since $\rho_v(g) = \prod_i (\rho_{w_i}(g))^{r_i}$, equivariance of the numerical coefficient follows
\begin{equation}
    \nabla^\varphi Q = \rho_v(g)^{-1}\dd \Pi \otimes \mu^{(\varphi^g)}\,,
\end{equation}
familiar from gauge-covariant derivatives. Structurally, it should be clear that a local unit transformation is precisely a gauge transformation taking us between trivialisations $\varphi\mapsto \varphi^g$ and appropriately transforming the connection $A^{(\varphi)}\mapsto A^{(\varphi^g)}$. Indeed, $\nabla^\varphi Q$ is also itself a dimensionful (1-form valued) quantity, in the sense that it is a section of the weighted bundle $T^*M\otimes \mathcal{L}_v$.

\subsection{On scalar-tensor theories and frame covariance}\label{sec:scalartensor}

Although non-constant units have received comparatively little direct attention in the dimensional analysis literature, the closely related problem of spacetime variation of fundamental constants has been extensively discussed in cosmology \citep{Uzan2025}. As cosmologists, we are understandably tempted to ponder whether the reference scales and constants we rely on when it comes to physics here on Earth (the speed of light, Planck's constant, Newton's constant, and so on) have varied over cosmic timescales or distances --- where any variation might otherwise pass unnoticed in local laboratory settings. Toying with fundamental constants by, for example, promoting them to new dynamical scalars or allowing them to depend on vacuum expectation values of other fields, provides a wide arena in which to address long-standing problems in cosmology; including but not limited to questions surrounding dark matter, dark energy, inflation, and naturalness. This is nothing new. But since gravity plays a central role in these settings, perhaps an instinctive place to start and appreciate the physics is by granting dynamics to the gravitational coupling $\kappa=8\pi Gc^{-4}$ by replacing $\kappa^{-1}\mapsto f(\phi)$ where $\phi$ is some new scalar. Doing so for standard Einstein-Hilbert gravity coupled to matter yields, in the simplest example, a scalar-tensor theory expressed in the so-called Jordan frame,
\begin{equation}\label{eq:Jordanframe}
    S = \int \dd^4x \sqrt{-g_\mathsf{J}} \frac{f(\phi)}{2}R_\mathsf{J} + S_\text{m}[g^\mathsf{J}_{\mu\nu};\psi]\,.
\end{equation}
These theories are often phenomenologically challenged because they introduce light propagating degrees of freedom and by fifth force constraints. In any event, it is maybe not immediately obvious that the action should face these problems, since it does not contain explicitly a kinetic term for the scalar $\phi$. This is easier to realise after performing a Weyl transformation
\begin{equation}
    g^\mathsf{J}_{\mu\nu}\longmapsto g^\mathsf{E}_{\mu\nu} = \Omega^2g^\mathsf{J}_{\mu\nu}
\end{equation}
with $\Omega^2\in C^\infty(M,\R_+)$ some function. For the particular choice $\Omega^2 = f/M^2$ where $M^2$ is some constant with the same dimension as $f$, the action is put in Einstein frame
\begin{equation}\label{eq:Einsteinframe}
    S = \int \dd^4x \sqrt{-g_\mathsf{E}} \left\{ \frac{M^2}{2}R_\mathsf{E}- \frac{3}{4}M^2\left(\frac{f'(\phi)}{f(\phi)}\right)^2 g_\mathsf{E}^{\mu\nu}\partial_\mu \phi\partial_\nu \phi\right\} + S_\text{m}[\Omega^{-2}g^\mathsf{E}_{\mu\nu};\psi]
\end{equation}
so that the non-minimal coupling has been exchanged with an explicit kinetic term and a modified matter coupling.

This slight digression on the Einstein and Jordan frame representations is to spotlight how Weyl transformations in gravitational theories inevitably lead us to confront non-constant units. Consider the line element $\dd s^2 = g_{\mu\nu} \dd x^\mu \dd x^\nu$ as a natural candidate for a spacetime ruler and measure of distances. In special relativity, for instance, we know that two observers linked by Lorentz transformations will see the same line element, allowing them to unambiguously share the same spacetime ruler. But, as we have just seen, the situation is more subtle once we look to general relativity. Under a Weyl transformation, the line element is no longer invariant, $\dd s^2\mapsto \Omega^2 \dd s^2$, and we are now nudged to think in terms of a new (spacetime-dependent) ruler that differs by a scale factor $\Omega$ \citep[see also][]{Eddington1921,Eddington1923}.

\citet{Dicke1962} realised that the passage between the Einstein and Jordan frame can effectively be viewed as --- and thereby also undone by --- an appropriate unit transformation. That is, working in units of the Planck mass $M$ in Einstein frame yields the same dimensionless ratios when working in units of the coupling $\sqrt{f}$ in Jordan frame.\footnote{This is further detailed in \citet{KaramitsosMuntz2025} \citep[see also][for a related discussion]{Bento2025}. We also highlight that there is really no such thing as \emph{the} Einstein or Jordan frame, as there exist infinitely many smooth families of conformal frames one can choose from. E.g.~all Einstein frames are related by the scale transformations that leave the connection \eqref{eq:connection} invariant.} The question of whether different conformal frames represent physically indistinguishable theories has been dubbed the \emph{frame problem}. Although it was arguably resolved by \citeauthor{Dicke1962} already in \citeyear{Dicke1962}, the debate somehow prevailed many years later \citep[see][Section~3]{FaraoniGunzigNardone1999}.\footnote{Specifically I am referring to the question of semiclassical equivalence. It remains an open question whether conformal frame equivalence also holds at the quantum level.} One attempt to untangle the situation is in the frame-covariant formalism of scalar-tensor theories \citep{Flanagan2004, KuuskJarvVilson2016, Karamitsos2018}. There it is noticed that the action can be cast in a way that is manifestly form-invariant under Weyl transformations,
\begin{equation}
    S[g,\phi,\dots] = S[g^\Omega,\phi^\Omega,\dots]\,.
\end{equation}
Superscripts denote the Weyl-transformed arguments and the ellipsis is a placeholder for other possible fields, couplings, or model functions of the theory. Form-invariance is therefore also reflected at the level of the equations of motion. This is made manifest by defining the `frame-covariant derivative' meant to replace ordinary derivatives,\footnote{Here $w$ denotes the conformal weight, where conventionally $X\mapsto \Omega^{-w}X$ under Weyl transformations. The metric $g\mapsto \Omega^2g$ thus has weight $-2$. Note that $X$ in the original definition is allowed to be tensorial. I suppress indices for ease of discussion.}
\begin{equation}
    D X \coloneqq \dd X - w(\dd \log\sqrt{f})X\,.
\end{equation}
We can check that it transforms covariantly under Weyl transformations,
\begin{equation}
\begin{split}
    D^\Omega X^\Omega &= \dd (\Omega^{-w}X) - w(\dd \log \sqrt{\Omega^{-2}f})(\Omega^{-w}X)\\
    &= \Omega^{-w}DX\,.    
\end{split}	
\end{equation}
This derivative is introduced in \citet{Karamitsos2018} on operational grounds by requiring derivatives to transform covariantly under frame transformations. Hopefully the reader will agree that this smells a lot like a gauge-covariant derivative. There are at the very least many analogies to be made between the frame-covariant formalism and gauge theory; these have not gone entirely unnoticed. \citet[p.~3]{QuirosDeArcia2018} for instance intuit ``Actually, following the spirit of the above examples: coordinate invariance of the laws of gravity in GR and gauge invariance of the laws of electromagnetism, one should require the action and the field equations of the theory -- representing the physical laws -- to be invariant under [Weyl transformations].'' More recently, \citet[p.~8]{JarvKaramitsos2026} motivate the same construction, emphasising its role as a representational tool for building frame-invariant objects: ``The definition of frame-covariant derivatives is useful because it gives us a recipe to construct invariant quantities even when derivatives are involved, and helps us keep track of conformal weights. ...~reminiscent as it is to a gauge-covariant derivative, [it] is a straightforward way to promote quantities to their covariant counterparts.''

What has to my knowledge not been explicitly argued in the literature is that this analogy can be made exact: the frame-covariant derivative \emph{is} literally obtained from the gauge-covariant derivative induced on the weighted bundles associated with the principal $\R_+$-bundle of local scales.\footnote{Note that the quantity $X$ should not be understood as a section of an associated bundle, but rather as the equivariant function related to it via Eq.~\eqref{eq:Pi_definition}. In our notation, for a weight-one quantity represented by $X$, the associated section $Q=X\otimes (\sqrt{f})^{-1}$. If $\nabla$ is adapted to the Einstein-frame Planck mass, then ${\nabla Q=DX\otimes (\sqrt{f})^{-1}}$. More precisely, it is nothing other than the unit-covariant derivative written in the particular gauge ${\mu = (\sqrt{f})^{-1}}$, and the term $-\dd \log\sqrt{f}$ is the local connection 1-form in this trivialisation.}

The bundle language makes this identification more precise. Consider the principal $\R_+$-bundle of conformal metrics $S^2T^*M\supset \mathcal{Q}\to \mathcal{G}\to M$ (c.f.~\cref{footnote:conformalgeometry}). I will assume that the relevant scale bundle has been identified with $\mathcal{Q}$ by an isomorphism of principal $\R_+$-bundles.\footnote{I henceforth assume $k=1$, which is usually taken to be the case in the study of scalar-tensor theories, where one works in natural units.} Then every weighted line bundle is associated to $\mathcal{Q}$. We may for instance write down the bundle of dimensionful metrics
\begin{equation}
    \ell^2g \in \Gamma(\mathcal{Q}_{-2}) \simeq  {}^{\displaystyle \left(\Gamma(\mathcal{Q})\times \Gamma(\mathcal{L}_{-2}^+)\right)} {\big\slash}_{\displaystyle \sim}
\end{equation}
identifying $(g,\ell^2)\sim (\Omega^2g,\Omega^{-2}\ell^2)$. If we want to consider a collection of other fields occupying the sum of associated vector bundles $\mathcal{V}=\bigoplus_\alpha \mathcal{V}_\alpha$, with each $\mathcal{V}_\alpha$ carrying some definite weight, then a field configuration is naturally a section $\Psi$ of the total field bundle $\mathcal{E}\coloneqq \mathcal{Q}_{-2}\oplus \mathcal{V}$. A scalar-tensor theory is then described by an appropriate functional
\begin{equation}
    S\colon \quad \Gamma(\mathcal{E})\to \R
\end{equation}
(i.e.~it takes a familiar form like \eqref{eq:Jordanframe}) built from the fields in $\mathcal{E}$, the spacetime differential, and the covariant derivatives induced from the principal connection on $\mathcal{Q}$. At the level of the dynamics defined by this action, the equations of motion are represented by a differential operator of order $n$. That is, they correspond to a bundle map
\begin{equation}
    \mathsf{D}\colon\quad J^n\mathcal{E}\to \mathcal{W}
\end{equation}
from the $n$\textsuperscript{th} jet bundle to some (weighted) target bundle $\mathcal{W}$, where the solution set consists of those field configurations $\Psi\in \Gamma(\mathcal{E})$ for which ${\mathsf{D}\circ j^n\Psi = 0}$. If $\mathsf{D}$ is $\R_+$-equivariant then the zero section in $\mathcal{W}$ is fixed by the group action,
\begin{equation}\label{eq:equivariantD}
    \mathsf{D}(j^n\Psi) = 0\qquad \implies \qquad \mathsf{D}(j^n(g\cdot \Psi)) = g\cdot \mathsf{D}(j^n\Psi) = 0\,,
\end{equation}
so $g\cdot \Psi$ is a solution whenever $\Psi$ is. $\R_+$-equivariance is precisely the property that fails for ordinary derivatives acting on non-zero-weight fields, and what is restored by promoting ordinary derivatives to covariant ones. However, recall the key insight highlighted earlier; the action functional $S$ is \mbox{(form-)invariant} under the local $\R_+$-action on $\Gamma(\mathcal{E})$, implying that its first variation as well as the Euler-Lagrange equations, given by such an operator $\mathsf{D}$, will be $\R_+$-equivariant.\footnote{Consider the simpler example of a charged massive scalar, $S\supset \int \dd^4x\,\phi^\dagger(\Box+m^2)\phi$. This part of the action is gauge-invariant only when $\Box$ is built from the corresponding gauge-covariant derivatives. The same goes for equivariance of the Klein-Gordon equation of motion $\mathsf{D}\phi=0$, where ${\mathsf{D}\colon \phi\mapsto (\Box+m^2)\phi}$ is viewed as a map. If $\Box$ was built from ordinary derivatives, the gauge group need not be an endomorphism on the space of solutions $\mathrm{ker}\mathsf{D}$.} The classical equivalence of conformal frames is then a simple corollary: the solution space descends to the quotient 
\begin{equation}
    \ker\mathsf{D}\longrightarrow\ker \mathsf{D}\slash C^\infty(M,\R_+)\,.
\end{equation}
Structurally, what this means is that solutions to the field equations in the Jordan and Einstein frame (and any other related conformal frame), respectively lie in the same $\R_+$-orbit in configuration space and describe the same dynamics up to local rescaling. The coordinates probing directions in $\ker \mathsf{D}$ transverse to the orbit are exactly the frame-invariant (gauge-invariant) combinations of fields into weight-zero objects. In \Cref{sec:buckingham} we shall see that this is just another manifestation of the \cref{thm:Buckingham}.

\section{On Dimensionlessness and Unit-independence}\label{sec:dimensionless}
In particle physics, we are accustomed to the fact that not all particles are born the same. Leptons and quarks, for instance, manifest distinct physical properties because of how they transform under different representations of the Standard Model gauge group. Adopting language to pinpoint these properties is incredibly helpful in daily discourse --- and fortunately physicists have standardised such nomenclature: jargon like `the electron is \emph{colourless}' is really just to say that the electron associated bundle transforms in the singlet representation of the $SU(3)_c$ subgroup (quite a mouthful in comparison), while `the Yukawa interaction term is \emph{gauge-invariant}' points out how this term is invariant under the full structure group $G_\text{SM}$. 

A similar story can be told about dimensional analysis. In the theory of measurement, it has long been recognised that not all quantities are born the same, in the sense that quantities can admit different classes of admissible transformations. The classic account is due to \citet{Stevens1946}, who classified scales of measurement into four types: nominal, ordinal, interval, and ratio. In contrast to particle physics, however, the corresponding terminology (and awareness of the fact) is far from established. \citeauthor{Stevens1946} exemplifies how discussions of measurement turn volatile once one conflates different empirical operations under colloquial umbrella terms such as `scale of measurement'. Without a structural foundation, it becomes almost too easy to sweep substantive distinctions under the semantic carpet. This motivates \citeauthor{Stevens1946} to offload part of the burden of articulation onto an underlying group structure. ``Perhaps agreement can better be achieved if we recognize that measurement exists in a variety of forms and that scales of measurement fall into certain definite classes.'' \citep[p.~677]{Stevens1946} Other terms like units, dimensions, and quantities face similar obfuscations. For example, while torque `has the dimension of energy' and can be expressed in the same units, we should not conclude from this alone that they are the same kind of quantity. As \citet[p.~L22]{Emerson2005} puts it: ``The `dimension' of a quantity tells us far less about it than its definition does, if indeed its dimension tells us anything at all.'', to which \citet[p.~120]{Hall2022} echoes this point when he writes ``Unfortunately, `quantity' and `dimension' have become hopelessly entangled.'' In the spirit of \citeauthor{Stevens1946}, let us therefore explore how gauge theory may help us codify some of these terms.\footnote{\label{footnote:gaugeinvariant}All that being said, colloquialisation is not entirely alien to gauge theory either. A mild version of the same phenomenon can be heard in particle physics, where one occasionally encounters loose statements like ``the neutrino is chargeless'', referring to the fact that the neutrino is electromagnetically neutral, transforming trivially under $U(1)_\text{em}\subset G_\text{SM}$. However, to the extent that particles may be `charged' under other gauge subgroups, the colloquial meaning of the statement does not intend to claim that the neutrino is gauge-invariant. Why does this not spark confusion? My guess is that the basic nomenclature has largely stabilised in particle physics, whereas in dimensional analysis one still regularly encounters semantic disputes at the level of the central terms themselves.}

Following \citeauthor{Stevens1946}'s classification, ratio scales involve quantities whose admissible numerical coefficients are related by pure scale transformations, $\Pi \mapsto \Pi ' = \rho\,\Pi$ with $\rho$ some conversion factor, when expressed in different units or representations of said quantity. They are the ones we most often encounter in physics: examples include mechanical quantities such as mass, length, and time. Another example is financial currencies.\footnote{In fact, the connection between currency conversion and gauge theory is already known in the literature. Building on \citet{Young1999}, \citet{Maldacena2015} actually uses the finance of currency conversion as a more pedagogical introduction to the concept of gauge theories. It is notably an $\R_+$ gauge theory.} What these structurally all have in common is that ratio scales admit a unit-independent zero point. That is, $\Pi=0$ always maps to the same zero in any other allowed units. We immediately recognise that admissible transformations of ratio quantities may be captured by the group $\R_+$. 

Interval scales, on the other hand, provide the simplest example where we must look beyond the group $\R_+$ of scale transformations. The canonical example of such a quantity is temperature. Conversion from Celsius to Kelvin or Fahrenheit can involve both a scaling and a shift.\footnote{I want to thank Caspar Jacobs for bringing this to my attention.}
\begin{equation}
    \Pi^{(\SI{}{\celsius})} \mapsto \Pi^{(\SI{}{\kelvin})} = \Pi^{(\SI{}{\celsius})} + 273.15\,,\qquad \Pi^{(\SI{}{\celsius})} \mapsto \Pi^{(\SI{}{\degree F})} = \frac{9}{5}\Pi^{(\SI{}{\celsius})} + 32
\end{equation}
In other words, scale transformations do \emph{not} exhaust the set of possible unit transformations for temperature. One must look to the affine group $\mathrm{Aff}(1,\R)\simeq \R\rtimes \R_+$.\footnote{That is at least before the adoption of the Kelvin scale, which distinguishes `absolute zero temperature', \SI{0}{\kelvin}, as \emph{the} zero point. In statistical physics, for example, one will rarely (if ever) work in units like Celsius or Fahrenheit, related to Kelvin via shifts, such as to discern zero temperature. Fixing an origin as such reduces $\mathrm{Aff}(1,\R)$ to its $\R_+$ subgroup, effectively treating temperature as a ratio quantity instead of an interval quantity.} 

The remaining two classes in \citeauthor{Stevens1946}'s taxonomy may be understood by looking to even larger groups. For ordinal quantities, only hierarchies or inequalities between representatives are preserved under a change of units. Hardness scales of minerals provide a standard example: passing from the Mohs to Vickers preserves order, but not differences or ratios. ``Diamond is harder than quartz'' is indeed a unit-independent statement, but there is no unit-invariant way of stating numerically \emph{how much} harder diamond is than quartz. It is therefore fitting to identify the admissible transformations of ordinal quantities with order-preserving diffeomorphisms $\mathrm{Diff}_+(\R)$. 

Nominal quantities are weaker still. Here, only sameness and difference survive. Library sorting systems, social security numbers, telephone area codes, football player jerseys, and other classificatory tags are all nominal. Although people in practice adopt particular tags with other conveniences in mind (for instance the first four digits of modern arXiv identifiers also refer to the date on which the article was uploaded), relabelling them does not alter the represented content. Nominal quantities might possibly take us beyond the physicist's typical notion of what a quantity entails. Even so, it is conceptually befitting to place it in the hierarchy of structure groups we have seen so far. More specifically, if we insist on using $\R$ as an auxiliary carrier of labels, we can model arbitrary relabelling by the full diffeomorphism group $\mathrm{Diff}(\R)$. I summarise this discussion in \Cref{table}.

From the standpoint of gauge theory, we therefore see that dimensional analysis is not intrinsically tied to the principal $\R_+^k$-bundle. That is merely the special case appropriate to ratio quantities. More generally, we may consider a principal $G$-bundle, where $G$ is the group of admissible transformations appropriate to the types of quantities under discussion. Interval, ordinal, and nominal quantities are obtained by choosing larger structure groups.

\begin{table}[H]
    \renewcommand{\arraystretch}{1.5}
    \centering
    \begin{tabular}{|c|c|c|}
        \hline \textbf{\textsf{Quantity type}} & \textbf{\textsf{Gauge group}} & \textbf{\textsf{Invariant}} \\ \hline
        Ratio & $\R_+$ & Ratios $\tfrac{q_1}{q_2}$\\ \hline
        Interval & $\mathrm{Aff}(1,\R)$ & Ratio of differences $\tfrac{\Delta q_1}{\Delta q_2}$\\ \hline
        Ordinal & $\mathrm{Diff}_+(\R)$ & Inequalities $q_1\leq q_2$\\ \hline
        Nominal & $\mathrm{Diff}(\R)$ & Equality $q_1= q_2$\\ \hline
    \end{tabular}
    \caption{A summary of different types of quantities, the corresponding admissible transformation group, and the kind of invariants one can build from its action. For interval, ordinal, and nominal scales these actions should be understood as actions on an associated fibre, not necessarily as linear representations.}
    \label{table}
\end{table}

Enlarging the structure group also highlights a crucial subtlety concerning the notion of a dimensionless quantity \citep{Emerson2005, Hall2022}. Consider, for instance, the ratio of two temperatures
\begin{equation}
    \frac{T_1}{T_2}\,.
\end{equation}
From common intuition of dimensional analysis we readily appreciate that the result will give us a dimensionless quantity, in the mundane sense that one would ordinarily write down no unit symbol (or at least append a formal `1'). Or, in the gauge theory language laid out in this paper, we may say more concretely that it has weight $w=0$ under scale transformations. Writing a `1' as the unit can however be misleading because it superficially tempts us to compare the result with ordinary numbers. Despite being dimensionless in the scaling sense, the ratio of two temperatures does not behave like a unit-independent number. Take $\SI{0}{\celsius}/\SI{5}{\celsius} = 0$ whereas when expressed in Fahrenheit $\SI{32}{\degree F}/\SI{41}{\degree F} \neq 0$. In contrast, a ratio of temperature differences
\begin{equation}
    \frac{\Delta T_1}{\Delta T_2} = \frac{T_1 - T_1'}{T_2-T_2'}
\end{equation}
is dimensionless \emph{and} will agree numerically no matter what units they are expressed in, precisely because the affine shifts cancel. Dimensionless ratios of \emph{interval} quantities evidently behave differently from dimensionless ratios of \emph{ratio} quantities: the former need not be unit-independent. The same observation of course extends to the remaining scale types. While the clash appears to be mainly semantic \citep{Hall2022}, the gauge theory picture allows us to make the separation more transparent:

\vspace{0.25cm}
\begin{description}[font=\textsf,labelindent=1cm]
    \item[Dimensionless:] invariant under scale transformations $\mathbb{R}_+^k\leq G$.
    \item[Unit-independent:] invariant under the full structure group $G$.
\end{description}
\vspace{0.25cm}

One can draw a straightforward parallel with the `colourless' versus `gauge-invariant' distinction explained at the start of this section (see also \cref{footnote:gaugeinvariant}). For instance, whilst a dimensionless quantity is not necessarily unit-independent, all unit-independent quantities must also be dimensionless. Only for ratio quantities do these two notions coincide (`chargeless' analogously coincides with `gauge-invariant' for pure electromagnetism). The fact that a word like `dimensionless' is so widespread in daily speech perhaps also reflects the over-abundance of ratio quantities in physics, which understandably makes the difference easy to miss. \citet{Roche1998} is one of the rare examples where unit-independence receives separate attention --- actually he arrives incredibly close to the intuition presented here, referring to the invariance under unit transformations as a ``gauge invariance''. This has not, to my knowledge, been defended elsewhere. On the contrary, \citet{Grozier2020} finds \citeauthor{Roche1998}'s nomenclature somewhat unsatisfactory. He elaborates ``because ``variation in unit size'' makes sense only after one has chosen a unit, whereas I see unit-invariance as a property independent of \emph{any} unit system'' \citep[p.~10]{Grozier2020}. While I concur that unit-independence as a property appears prior to specifying a set of base units, I will nevertheless argue that \citeauthor{Roche1998}'s use of the term is structurally accurate and aligns with \citeauthor{Grozier2020}'s sentiment. To see this, recall from \Cref{sec:gaugetheory} how adopting a set of base units is the same as trivialising the underlying principal bundle. It is only once a trivialisation is chosen that it becomes possible to parametrise how gauge transformations act on sections of associated bundles. \citeauthor{Grozier2020} is therefore right in claiming that to see ``variation in unit size'' we must specify the units in the first place. Structurally, however, `\emph{in}variation in unit size' makes sense \emph{before} specifying the units; invariant quantities (those transforming in the trivial representation) are completely indifferent to the choice of trivialisation. \citeauthor{Roche1998}'s terminology therefore correctly captures the structural point that unit-independence corresponds to invariance under the entire structure group and aligns with \citeauthor{Grozier2020} when he writes that invariance itself is not tied to any particular system of units.

On a final note, let me briefly comment on the interpreted attitude towards dimensionful quantities in relation to the physical significance of unit-independence. As a heuristic statement, it is a widely agreed upon view that physically meaningful expressions should not depend on the choice of unit. Indeed, if the choice of units is merely a conventional representation of quantities, it is not bold to opine that something more fundamental to the laws of nature should not be variant under these choices. Following this rule strictly, we should conclude that \emph{only} objects built out of invariant building blocks --- like those presented in \Cref{table} --- are physically meaningful. Philosophers have referred to this as the \emph{Invariance Principle}: a quantity is physically real only if it is invariant under the symmetries of our theory \citep{MollerNielsen2017,Jacobs2021}. This is not particularly uncontroversial as a first approximation, but it does beg the question: where does this leave quantities themselves, which do generally vary under change of scale? Or, for that matter, relations like $F=ma$ built out of variant quantities, which the average physicist most likely will appreciate as physically meaningful. \citet{Jacobs2021} has argued that adopting the Invariance Principle for dimensionful quantities might read a bit too restrictive. \citet{Luce1978} and \citet{NarensLuce1987} are helpful comparisons. Their notion of meaningfulness is that the truth or falsity of a statement is preserved across all its possible representations; what \citeauthor{Jacobs2021} defends as a criterion for \emph{sophistication} \citep{Dewar2019}, which rejects the Invariance Principle but retains that symmetry-related models represent the same state of affairs. Translated into the principal bundle language, the sophisticated interpretation of gauge theories thereby accepts that configurations related by gauge transformations are physically equivalent (provided they are isomorphic), but does not require every physically meaningful object or statement to be gauge-invariant \citep{Jacobs2023,Gomes2026}. In particular, meaningfulness can also extend to \emph{equivariant} objects and statements. Consider the map $\mathsf{F}\colon \mathcal{L}_w\to \mathcal{L}_{w'}$. If $\mathsf{F}$ is equivariant with respect to the group action, then applying a gauge transformation to the statement $\mathsf{F}=0$ gives $\rho_{w'}(g)^{-1}\mathsf{F}=0$ and preserves its truth value across every gauge-related representative (c.f.~Eq.~\eqref{eq:equivariantD}). There is hence no requirement to restrict meaningfulness to exclusively unit-independent quantities in our account. On a sophisticated reading, quantities are themselves physically meaningful despite transforming non-trivially under unit-transformations, and may still enter into meaningful relations whose content is preserved across all choices of units. Compare this again to ordinary gauge theory. As particle physicists, it would also appear somewhat puzzling if we adopted the stance that only gauge-invariant objects are physically meaningful when the rudimentary building blocks (fields and their corresponding particles) are themselves \emph{variant} under gauge. Although I do not want to draw any broad conclusions here about the metaphysics of gauge theory and scales of measurement, I do find it relevant to draw the physicist reader's attention to the sophisticated attitude that appears motivated in the present setting.

\section{The Buckingham-$\Pi$ Theorem}\label{sec:buckingham}
Having discussed several aspects of dimensional analysis and how it may be appropriately viewed as a gauge theory, we now arrive at perhaps the most famous theorem in dimensional analysis.

\begin{theorem}[Buckingham-$\Pi$]\label[buckingham]{thm:Buckingham}
	Any law $F(Q_1, \dots, Q_m) = 0$ expressed in terms of $m$ quantities $Q_1,\dots,Q_m$ of $k$ base dimensions can be rewritten as a law $f(\Pi_1 , \dots, \Pi_{m-k}) = 0$ in terms of $m-k$ dimensionless quantities $\Pi$.
\end{theorem}

Since \citeauthor{Buckingham1914}'s \citeyear{Buckingham1914} paper, the now eponymous \cref{thm:Buckingham} has occupied a privileged place in the subject. Historically, it crystallised a line of thought already present in \citeauthor{Fourier1878}'s principle of dimensional homogeneity \citep{Fourier1878} and in the later work of \citet{Rayleigh1915}, \citet{Ehrenfest1926}, and \citet{Bridgman1931}, that if a physical equation is to make sense independently of our choice of units, then its form is strongly constrained before any detailed dynamics has been solved \citep{Jalloh2024, Jalloh2026}.

What makes the theorem so useful in practice is possibly also what made dimensional analysis appear slightly peculiar as a subject. On the one hand, the methodology most physicists will have practised --- listing all relevant quantities, noting their dimensions, and reducing the problem to a smaller set of dimensionless combinations --- is straightforward and works incredibly well to constrain the algebraic relations among physical quantities. On the other hand, the fact that this works at all is already telling us something structural about physical laws. For that reason the theorem has been so central not only in physical sciences, but also in philosophical discussions of similarity and the explanatory role of dimensions \citep{Sterrett2009, Sterrett2021, Jalloh2025}. 

For our purposes, however, I do not want to re-prove the theorem merely for the sake of supplying yet another proof. Many roads lead to Rome, and several of them are already well paved \citep{Ehrenfest1916, Gibbings1982, Gibbings2011}. The point of this section is instead to convey a different flavour of the theorem; one inspired by the gauge theory perspective adopted in this paper. Following \Cref{sec:dimensionless}, it is more transparent that the \cref{thm:Buckingham} does not necessarily have anything to do with special properties of quantities or how to engineer laws of physics. Rather, it is a statement about counting \emph{invariants}.\footnote{Also following the discussion in \Cref{sec:dimensionless}, it should be clear that the theorem mainly concerns properties of ratio quantities. Strictly speaking, \Cref{thm:Buckingham} does not require $f$ to be invariant, since the quantities $\Pi$ need not be unit-independent; only dimensionless. $f=0$ is therefore only guaranteed invariant under the $\mathbb{R}_+$ action.} Let us unpack: suppose we have some physical law or statement
\begin{equation}\label{eq:physicallaw}
    \mathsf{F}(Q_1,\dots,Q_m)=0
\end{equation}
where $Q_i\in \Gamma(\mathcal{L}_{w_i})$ are quantities of weight $w_i$. The tuple of quantities $(Q_1,\dots,Q_m)$ is then a section of the sum
\begin{equation}
    \mathcal{E}=\mathcal{L}_{w_1}\oplus \cdots \oplus \mathcal{L}_{w_m}\,.
\end{equation}
For instance, in a textbook dimensional analysis problem where we may be dealing with $m$ real quantities, the typical fibre $V$ of $\mathcal{E}$ is $V\simeq \underbrace{\R\oplus \cdots \oplus \R}_{\text{$m$ times}}$. Locally, the fibre carries the diagonal action
\begin{equation}
    g\cdot (q_1,\dots,q_m)=(\rho_{w_1}(g)^{-1}q_1,\dots,\rho_{w_m}(g)^{-1}q_m)\,.
\end{equation}
The map $\mathsf{F}$ then picks out a subspace $\mathcal{E}_\mathsf{F}\coloneqq \mathsf{F}^{-1}(0)\subset \mathcal{E}$ that obeys Eq.~\eqref{eq:physicallaw}. However, we are not interested in just any type of equation: in dimensional analysis (and by analogy in gauge theory) we are interested in equations whose truth value remains the same under rescaling of units (gauge transformations); c.f.~the discussion at the end of \Cref{sec:dimensionless} or \Cref{sec:scalartensor} for comparison. As a local statement at the level of the fibre, equations of interest are precisely the ones whose (local) solution space descends to the quotient over the typical fibre, $V/G$. That is, a map $\mathsf{F}$ for which $(Q_1,\dots,Q_m)\in \mathcal{E}_\mathsf{F}$ implies $g\cdot(Q_1,\dots,Q_m)\in \mathcal{E}_\mathsf{F}$. 

It is then a fact that (on a regular stratification) $G$-invariant functions on $V$ depend only on coordinates on the orbit space $V\slash G$.\footnote{The orbit space $V\slash G$ may fail to be Hausdorff when some orbits are not closed. Consider for example $\R_+$ acting on $\R^2$ by $g\cdot (x,y)=(gx,g^{-1}y)$. The orbit of $(x,0)$ with $x>0$ has $(0,0)$ in its closure, however the origin is its own distinct orbit. $\R^2\slash \R_+$ is therefore non-Hausdorff. That being said, I would not think this to be a problem for Buckingham-$\Pi$, since it is intended as a statement about generic quantities away from ill-behaved loci such as this; e.g.~where certain quantities vanish.} The local dimension of the quotient is
\begin{equation}
    \mathrm{dim}(V\slash G)=\dim(V)-\dim(G)+\dim(G_\text{stab})\,,
\end{equation}
where $G_\text{stab}$ is the generic stabiliser. When the group acts freely then the last term vanishes. It becomes simple to see why we should recover the \cref{thm:Buckingham}: if there are $m$ real-valued ratio quantities whose weights span the $k$ base dimensionalities, we may identify $V=\R^m$ and $G=\R_+^k$, and thus $\dim(V\slash G)=m-k$.

Let us look at some illustrative examples:
\paragraph{Rotational symmetry:} consider $\mathrm{SO}(3)$-invariant functions over $\R^3$, with the group acting by rotations. The stabiliser of any non-zero vector in $\R^3$ is $\mathrm{SO}(2)$, so away from the origin we find that $\dim(\R^3\slash \mathrm{SO}(3))=3-3+1=1$. In other words, such functions depend only on one variable (the radius).

\paragraph{Simple pendulum:} the simple pendulum is a textbook example in dimensional analysis. The task is to find a general expression relating various parameters of the system. A preliminary physics guess is that the period of the pendulum $T$ may depend on its length $\ell$, the gravitational acceleration $\gamma$ (to avoid conflating with group elements $g$), and mass $m$. We identify three base dimensionalities: mass, length, time (all of the ratio type). As a gauge theory, the structure group is therefore $G=\R_+^3$ and acts on the fibre $V\simeq\R^4$ where the quantities take their values. Assuming the weights span $\mathfrak{g}^*$,\footnote{If too many weights are degenerate, there is a possibility that the stabiliser becomes non-trivial as the weights fail to span all of $\mathfrak{g}^*$. E.g.~if all weights are proportional, ${\dim(\R^4\slash\R_+^3)=4-3+2=3}$.} one readily finds $\dim(\R^4\slash \R_+^3)=1$. Any unit-independent function over the four quantities must therefore locally be expressible as a function of one invariant combination,
\begin{equation}
    \Pi = \prod_{Q\in \{T,\ell,\gamma,m\}} Q^{r_Q}\qquad\text{such that}\qquad \sum_Q r_Qw_Q=0
\end{equation}
for some set of $r_Q\in \R$. 

Note how this approach does not directly tell us what the invariant $\Pi$ is. That is a subsequent task requiring us to solve the set of equations involving the exponents $r_Q$ and weights $w_Q\in\mathfrak{g}^*$. Suppose we pick a ``mechanical'' basis on $\mathfrak{g}^*$, such that the weight vector components read as follows:
\begin{table}[H]
    \centering
    \begin{tabular}{c*{3}{c}}
         & Mass & Length & Time\\ \hline
        $w_T$ & $0$ & $0$ & $1$\\
        $w_\ell$ & $0$ & $1$ & $0$\\
        $w_\gamma$ & $0$ & $1$ & $-2$\\
        $w_m$ & $1$ & $0$ & $0$\\
    \end{tabular}
\end{table}
One finds the unique solution $\Pi = T^2\ell^{-1}\gamma$ up to an overall power. The simple pendulum is therefore described by some unit-invariant equation $\mathsf{f}(\Pi)=0$. As always, symmetries can only take us so far when it comes to fixing the precise law or dynamics of the system --- only through further physical insight or more detailed calculations, e.g.~by explicitly solving the equations of motion, would we find for the simple pendulum (in the small-angle approximation) that $\mathsf{f}$ takes the form
\begin{equation}
    \mathsf{f}(\Pi) = \sqrt{\Pi}-2\pi=0\,.
\end{equation}

\paragraph{A ``$U(1)$ simple pendulum'':} in the spirit of \Cref{sec:U(1)}, let us explore how the same dimensional analysis exercise carries over to gauge theory. Consider a $U(1)^3$ gauge theory describing four charged complex scalars $\{\Psi_i\}_{i=1,\dots,4}$. A relevant exercise is to find the most general (local) gauge-invariant function written in terms of the four fields (for instance a scalar potential $\mathsf{V}(\Psi_i)$). Since the fields take values in $V\simeq \C^4$ (and assuming their weights span $\mathfrak{w}^*$), we compute ${\dim_\R(V\slash G)=8-3=5}$. Any gauge-invariant function can therefore be written in terms of the real and imaginary part of the complex singlets
\begin{equation}
    \Pi = \prod_{i=1}^4\Psi_i^{r_i}\bar{\Psi}_i^{r'_i} \qquad\text{such that}\qquad \sum_i (r_i-r_i')w_i=0\,.
\end{equation}
This is just the $U(1)$-analogue of the simple pendulum, with additional real invariants appearing because the fields are complex. Once again, this exercise alone is not sufficient in order to solve for the expression of $\Pi$. Suppose that the charges took the following values in the canonical basis of $\mathfrak{w}^*\simeq \Z^3$:
\begin{table}[H]
    \centering
    \begin{tabular}{c*{3}{c}}
         & $U(1)_1$ & $U(1)_2$ & $U(1)_3$\\ \hline
        $w_1$ & $0$ & $0$ & $1$\\
        $w_2$ & $0$ & $1$ & $0$\\
        $w_3$ & $0$ & $1$ & $-2$\\
        $w_4$ & $1$ & $0$ & $0$\\
    \end{tabular}
\end{table}
Four of these invariants are the norms $\Pi_i=\Psi_i\bar{\Psi}_i$, while the remaining invariant is the less trivial complex combination $\Pi_5=\Psi_1^2\bar{\Psi}_2\Psi_3$, analogous to $T^2\ell^{-1}\gamma$ in the previous example. Although na\"ively this yields six real invariants (two coming from the real and imaginary part of $\Pi_5$), they are subject to one constraint, $\Pi_5\bar{\Pi}_5=\Pi_1^2\Pi_2\Pi_3$. Thus we find that the quotient has the expected five real dimensions. A general gauge-invariant function can thus be written as $\mathsf{f}(\Pi_1,\dots,\Pi_4,\mathrm{Re}\,\Pi_5, \mathrm{Im}\,\Pi_5)$ subject to this relation.

\paragraph{The Standard Model Yukawa sector:} let us look at a slightly more involved example familiar from particle physics. In the quark sector of the Standard Model, the gauge and kinetic terms admit a global flavour symmetry $G=U(N)_{Q_L}\times U(N)_{u_R}\times U(N)_{d_R}$, where $N$ is the number of generations. This symmetry is broken by the Yukawa interactions
\begin{equation}
    \mathcal{L}_\text{Yukawa} = - \bar{Q}_LY_u\widetilde{H}u_R - \bar{Q}_LY_dHd_R + \text{h.c.}
\end{equation}
where $Y_u,Y_d\in \C^{N\times N}$. Equivalently, however, one may retain the flavour symmetry by regarding the Yukawa matrices as spurions transforming according to
\begin{equation}
    Y_u\mapsto V_QY_uV_u^\dagger\,,\qquad Y_d\mapsto V_QY_dV_d^\dagger
\end{equation}
for $(V_Q,V_u,V_d)\in G$. The problem of counting the number of independent physical parameters controlling the Yukawa couplings then becomes a problem of determining the dimension of the orbit space over the pair $(Y_u,Y_d)$. The only generic stabiliser is $V_Q=V_u=V_d=e^{i\alpha}1_N$  (corresponding to the baryon number $U(1)_B$) acting trivially on $Y_u$ and $Y_d$. Hence the orbit space has dimension 
\begin{equation}
    \dim_\mathbb{R}((\C^{N\times N}\oplus \C^{N\times N})\slash U(N)^3)=4N^2-3N^2+1=N^2+1\,.
\end{equation}
In other words, the Yukawa matrices contain the information of $N^2+1$ independent physical parameters. It is well-known that these comprise the $2N$ quark masses after symmetry breaking, $\tfrac{N(N-1)}{2}$ mixing angles, and $\tfrac{(N-1)(N-2)}{2}$ CP-violating phases. 

Finding all the invariant combinations is once again a supplementary exercise \citep[see][Section~5]{Jenkins2009}. Define ${X_u\coloneqq Y_uY_u^\dagger}$ and ${X_d\coloneqq Y_dY_d^\dagger}$. The relevant flavour action then acts via simultaneous conjugation
\begin{equation}
    X_u\longmapsto V_QX_u V_Q^\dagger\,,\qquad X_d\longmapsto V_QX_d V_Q^\dagger\,. 
\end{equation}
For $N=3$ in the Standard Model, it turns out that there are ten algebraically independent CP-even invariants 
\begin{equation}
\begin{gathered}
    \tr X_u \qquad \tr X_d\\
    \tr X_u^2 \qquad \tr X_uX_d \qquad \tr X_d^2\\
    \tr X_u^3 \qquad \tr X_u^2X_d \qquad \tr X_uX_d^2 \qquad \tr X_d^3\\
    \tr X_u^2X_d^2
\end{gathered}
\end{equation}
The remaining CP-odd invariant was found by \citet{Jarlskog1985},
\begin{equation}
    J\propto i\tr [X_u,X_d]^3\,.
\end{equation}
Its square yields a polynomial in the ten CP-even invariants. Just like in the previous example, despite having na\"ively discovered more than ${N^2+1=10}$ invariants expected from dimension counting, the surplus is accounted for by algebraic constraints. The most general $G$-invariant polynomial is accordingly of the form\footnote{\label{footnote:syzygies}Coordinates on the orbits may thus not suffice to generate the $G$-invariant algebra, since additional invariants can be required to distinguish different branches or orientations. For example, let $U(1)$ act diagonally on $(x,y)\in\C^2$. A convenient set of real polynomial invariants is $\abs{x}^2,\abs{y}^2, \mathrm{Re}(\bar{x}y),\mathrm{Im}(\bar{x}y)$ subject to the relation $\mathrm{Re}(\bar{x}y)^2+\mathrm{Im}(\bar{x}y)^2 = \abs{x}^2\abs{y}^2$. This is consistent with $\mathrm{dim}_\R(\C^2\slash U(1))=3$. The orbit can be parametrised by two magnitudes and one relative angle $\theta$. However, if one only keeps $\mathrm{Re}(\bar{x}y)$, one cannot distinguish between relative phases $\theta$ and $-\theta$; the phase-odd invariant $\mathrm{Im}(\bar{x}y)$ is also required to explore the full $G$-invariant algebra.}
\begin{equation}
    \mathsf{f}_1(\Pi_\text{CP-even}) + J\mathsf{f}_2(\Pi_\text{CP-even})\,.
\end{equation}

I hope it is clear from the examples above that the gist of the \cref{thm:Buckingham} is not unique to dimensional analysis, and that it generalises to much broader applications in gauge theory and particle physics as well. Counting and determining invariant combinations is evidently not a specialist task in dimensional analysis. It is quite the opposite: there is an entire field in mathematics dedicated to it, known as \emph{invariant theory}.

Suffice it to say, the overarching goal of classical invariant theory is to identify the properties of mathematical objects that remain invariant under a specified group action. Given a group $G$ acting linearly on a finite-dimensional vector space $V$ over a field $\Bbbk$, one can study the properties of the space of $G$-invariant polynomials $\Bbbk[V]^G$. For instance, Hilbert famously proved for a wide range of cases that $\Bbbk[V]^G$ is a finitely generated algebra over $\Bbbk$. From the examples above we can imagine how determining all generators and relations (syzygies) among them is generally a very complicated task. For a general compact Lie group, invariant-building becomes much harder because fields usually sit in non-trivial representations and involve contracting fields into singlets using various available invariant tensors, such as Kronecker deltas, Levi-Civita symbols, traces, etc. These contractions can satisfy non-trivial algebraic relations, so the number of invariant generators need not coincide with the dimension of the orbit space (c.f.~the Yukawa example and \cref{footnote:syzygies}).

It should be highlighted, though, that applications of invariant theory to gauge theory and particle physics have been studied for many years. Suppose, for example, we are tasked to build the most general extension of the Standard Model by including higher-dimensional operators, as in the Standard Model Effective Field Theory (SMEFT) or Higgs Effective Field Theory (HEFT). There are many ways to combine and contract fields into Lorentz- and gauge-invariant combinations. At low mass dimension this can be done by hand. For higher mass dimension, however, the combinatorics quickly become unmanageable if one also takes into account redundancies due to identities, integration by parts, and so on. Counting and organising the algebraically independent operators at each order in the EFT expansion is therefore naturally a problem for invariant theory. Hilbert series methods have become a systematic way to attack this problem and are extensively explored in \citet{Jenkins2009, Hanany2011, LehmanMartin2015, LehmanMartin2016, Henning2016} as well as in later literature. 

All this begs a natural question. If the \cref{thm:Buckingham} is an $\mathbb{R}_+^k$ instance of a more generic invariant-counting problem, is there an analogous statement for other groups --- possibly ones familiar from particle physics? The answer is \emph{yes}. A beautiful theorem due to \citeauthor{Schwarz1975} goes as follows:

\begin{theorem}[\citet{Schwarz1975}]
    \label[schwarz]{thm:schwarz}
    Let $V$ be a real finite dimensional representation of a compact Lie group $G$. Let $\Pi_1,\dots,\Pi_n$ be generators of $\R[V]^G$, the algebra of $G$-invariant polynomials on $V$. Each $G$-invariant $C^\infty$-function $F$ on $V$ has the form $F=f(\Pi_1,\dots,\Pi_n)$ for some $C^\infty$-function $f$ on $\R^n$.
\end{theorem}

\cref{thm:schwarz} on invariant functions is the cousin of the \cref{thm:Buckingham} for compact Lie groups, although it does not predict the number $n$ of invariant generators. For a general compact Lie group $G$, determining $n$ is a central problem of classical invariant theory (and therefore not one I shall attempt to solve here).

From the present discussion I hope to have conveyed that the \cref{thm:Buckingham} and its application to dimensional analysis are but an unusually tractable corner of a much larger invariant-theoretic problem. In the case of ordinary dimensional analysis the solution is relatively straightforward. In particle physics, on the other hand, the structure is usually much richer. That these methods and theorems are already woven into the technology used in modern gauge theory is a curious hint that the gauge-theoretic reading of dimensional analysis was never actually far away. Perhaps they were only ever speaking different dialects of the same language.

\newpage

\section{Conclusion}
In \citeyear{Weyl1918}, \citelinktext{Weyl1918}{Hermann Weyl} introduced the world to what is now retrospectively regarded as the first gauge theory. His original model --- an attempt to unify gravity and electromagnetism --- did not involve the compact Lie groups we usually associate with modern particle physics. Rather, it was built on the idea of \emph{scale} symmetry, with $\R_+$ playing the role of the gauge group. Indeed, \citeauthor{Eddington1921} already made explicit at the time the intuition that \citeauthor{Weyl1918}'s gauge transformations could be understood as local changes of units (\citealp{Eddington1921}; see also \citealp[Sec.~85--86]{Eddington1923}). \citeauthor{Weyl1918}'s model, however, did not achieve its intended physical unification, and the local scale transformations soon gave way, in what would become the mainstream continuation of gauge theory, to local phase transformations during the late-1920s development of quantum mechanics.\\
I find the history particularly noteworthy. At roughly the same time, the early twentieth century saw an active debate concerning the methodology and metaphysical foundations of dimensional analysis \citep{Jalloh2024}. I will not speculate why these circles did not merge their ideas, but it is useful to observe that the two traditions developed in different disciplines: gauge theory in relativity and field theory, and dimensional analysis in more classical settings such as mechanics and hydrodynamics. Even so, they were never very far apart. This becomes especially clear when considered in the same light, as later made explicit by \citet{Dicke1962}, who placed unit transformations and Weyl transformations on the same footing. The connection is easy to miss, not least because the habitual use of natural units in theoretical physics today tends to hide the choice of scale.

One purpose of the present paper has been to assemble the pieces of this puzzle. It has long been unclear what mathematical foundation underlies dimensional analysis. Yet, as a methodology, dimensional analysis is something most physicists learn early in their careers --- it is simple, intuitive, and oftentimes incredibly powerful. A mathematical scaffolding should not merely accommodate these qualities, but strengthen and extend them. On structural grounds, and with these perspectives in mind, I defended the claim that dimensional analysis can quite literally be understood as a gauge theory. In the simplest case of ratio quantities, it is equipped with structure group $\R_+^k$, where $k$ is the number of base dimensionalities. Dimensional quantities then appear as analogues of charged fields, as associated bundles transforming under different representations of the structure group. Quantity calculus also follows immediately from basic representation theory. Furthermore, I showed that any complete set of base units corresponds to a trivialisation of the underlying principal bundle; in this precise sense, a choice of base units is a choice of gauge.

The same framework also helped disentangle separate several subtleties related to scales of measurement. In \Cref{sec:scalartensor}, I clarified the physical equivalence between conformal frames in scalar-tensor theories from a bundle perspective. In particular, frame transformations are structurally equivalent to, and can thus be undone by, appropriate unit transformations, echoing arguments by \citet{Dicke1962} and the frame-covariant formalism (\citealp{Flanagan2004, KuuskJarvVilson2016, Karamitsos2018}; see also \citealp{KaramitsosMuntz2025}). In \Cref{sec:dimensionless}, I argued that different types of scales of measurement can also be encoded in larger structure groups. This, in turn, clarifies a semantic challenge in distinguishing dimensionless from unit-independent quantities \citep{Emerson2005, Hall2022}: the former are invariant under $\R_+^k$, whereas the latter are invariant under the full structure group $G\geq \R_+^k$.

Throughout this paper, I identified several ways in which dimensional analysis and gauge theory illuminate one another. Drawing analogies with $U(1)^k$ gauge theory, for example, gives a useful picture of how number-unit decomposition of quantities parallels the magnitude-phase decomposition of complex fields. On the other hand, $\R_+^k$ gauge theory is structurally less rich: in the setting considered here, it has no charge quantisation, and has trivial topology over paracompact bases. Moving from dimensional analysis back into gauge theory, I reinterpreted the \cref{thm:Buckingham} as a statement about gauge-invariant functions on representation spaces. Several pedagogical examples were used to demonstrate how the theorem carries over to similar counting problems in gauge theory. Notably, \cref{thm:schwarz} is the analogue of the \cref{thm:Buckingham} for compact Lie groups. Counting the minimum number of $\Pi$'s for a general compact Lie group and representation, however, is a difficult open problem of invariant theory in mathematics. Taken together, these perspectives strongly suggest that dimensional analysis and gauge theory are structurally equivalent formalisms: both rely on the same symmetry and covariance principles.

Allow me to close with a speculative observation concerning the role of dimensional analysis and units in fundamental physics. Throughout the present work, we have observed that the symmetry groups associated with scales of measurement are generally non-compact. If units are to be taken seriously as reflecting an internal gauge symmetry, on a par with the gauge groups of particle physics, this may have broader implications for the foundations of the theory. For example, there is an argument going back to \citet{BanksSeiberg2011} \citep[see also][]{GarglianoTudball2026} that effective theories obtainable from quantum gravity should not possess any non-compact gauge symmetries. If this adventuresome line of thought applies, it would suggest that effective theories of quantum gravity must fundamentally lack dimensionality; the respective gauge symmetries cannot be physical or dynamical, and are at best representational. An alternative possibility is that local changes of units must be tied to isomorphic symmetries that are not internal, such as Weyl transformations. That this framework potentially allows such questions even to be formulated precisely is, in my view, an inspiring prospect.

\section*{Acknowledgements}
I would like to thank Ali Bakhordarian, Ed Copeland, Zongzhe Du, Erik Søndergaard Gimsing, Henrique Gomes, Caspar Jacobs, Tony Padilla, and Kieran Wood for many helpful discussions and comments on the draft. I reserve a special thanks to Sotirios Karamitsos for many long discussions on dimensional analysis and units, and for bringing my attention to this topic. 
%I would especially like to thank Sotirios Karamitsos for many long discussions on the topic and for his influence on my interest in dimensional analysis.

%%%%%%%%%%%%%%%%%%%%%% BIBLIOGRAPHY %%%%%%%%%%%%%%%%%%%%%%

%\newpage
\begin{small} %%Makes bib small text size
\singlespacing %%Makes single spaced
\end{small} %%End makes bib small text size

\end{document}